\documentclass{llncs}

\newif\iffsen

\fsenfalse 

\usepackage{tikz,url,wrapfig,times}
\usetikzlibrary{arrows,automata,shapes,snakes}

\usepackage{listings,float}
\usepackage{amsmath}
\usepackage{moreverb}

\newcommand{\ie}{\emph{i.e.}}
\newcommand{\eg}{\emph{e.g.}}

\usepackage{xspace}
\newcommand{\whenphase}{\textit{when} \textit{phase}\xspace}
\newcommand{\actionphase}{\textit{action} \textit{phase}\xspace}
\newcommand{\whenclause}{\textit{when} \textit{clause}\xspace}
\newcommand{\whenclauses}{\textit{when} \textit{clauses}\xspace}
\newcommand{\actionclause}{\textit{action} \textit{clause}\xspace}
\newcommand{\actionclauses}{\textit{action} \textit{clauses}\xspace}

\newcommand{\FF}{\ensuremath{\mathcal{F}}\xspace}
\newcommand{\MM}{\ensuremath{\mathcal{M}}\xspace}

\floatstyle{boxed}
\newfloat{codelisting}{h}{clst}
\floatname{codelisting}{Listing}

\newfloat{trans}{h}{trans}
\floatname{trans}{Translation}

\begin{document}

\title{Analysing the Control Software of the\\
Compact Muon Solenoid Experiment\\ at the Large Hadron Collider}
\author{Yi-Ling Hwong\inst{1}\thanks{This work has been supported in part by a Marie Curie Initial Training
Network Fellowship of the European Community's Seventh framework program
under contract number (PITN-GA-2008-211801-ACEOLE).}
 \and Vincent J.J. Kusters\inst{1,2}
\and Tim A.C. Willemse\inst{2}}
\institute{
CERN, European Organization for Nuclear Research,\\
CH-1211 Geneva 23, Switzerland
\and
Department of Mathematics and Computer Science,\\
Eindhoven University of Technology\\
P.O.~Box~513, 5600~MB Eindhoven, The Netherlands
}

\maketitle

\begin{abstract}
  The control software of the CERN Compact Muon Solenoid experiment
  contains over 30,000 finite state machines. These state machines are
  organised hierarchically: commands are sent down the hierarchy and
  state changes are sent upwards. The sheer size of the system makes it
  virtually impossible to fully understand the details of its behaviour at
  the macro level. This is fuelled by unclarities that already exist at the
  micro level.  We have solved the latter problem by formally describing
  the finite state machines in the mCRL2 process algebra. The translation
  has been implemented using the \emph{ASF+SDF meta-environment}, and
  its correctness was assessed by means of simulations and visualisations
  of individual finite state machines and through formal verification of
  subsystems of the control software.  Based on the formalised semantics
  of the finite state machines, we have developed dedicated tooling
  for checking properties that can be verified on finite state machines
  in isolation.

  \end{abstract}

\section{Introduction}

The Large Hadron Collider (LHC) experiment at the European Organization
for Nuclear Research (CERN) is built in a tunnel 27 kilometres in
circumference and is designed to yield head-on collisions of two proton
(ion) beams of 7 TeV each.  The Compact Muon Solenoid (CMS) experiment
is one of the four big experiments of the LHC. It is a general purpose
detector to study the wide range of particles and phenomena produced
in the high-energy collisions in the LHC. The CMS experiment is made
up of 7 subdetectors, with each of them designed to stop, track or
measure different particles emerging from the proton collisions. Early 2010,
it achieved its first successful 7 TeV collision, breaking
its previous world record, setting a new one.

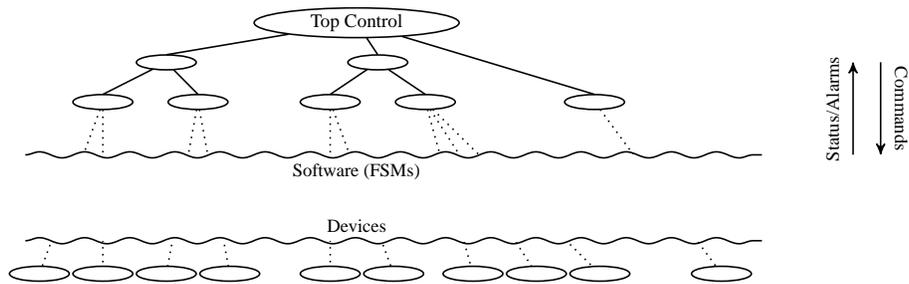
\begin{figure}[h]
\noindent
\begin{tikzpicture}[>=stealth',draw,node distance=60pt, semithick, inner sep=0pt, inner xsep=20pt, minimum size=0pt]

\node[xscale=0.5,yscale=0.5,state,shape=ellipse, ] (l0) {\Large Top Control};

\node[xscale=0.4,yscale=0.25,state,shape=ellipse] (k0) [below of=l0, xshift=-180pt] {};
\node[xscale=0.4,yscale=0.25,state,shape=ellipse] (k1) [below of=l0, xshift=20pt] {};

\node[xscale=0.4,yscale=0.25,state,shape=ellipse] (m0) [below of=k0, xshift=-60pt] {};
\node[xscale=0.4,yscale=0.25,state,shape=ellipse] (m1) [below of=k0, xshift=30pt] {};
\node[xscale=0.4,yscale=0.25,state,shape=ellipse] (m2) [below of=k1, xshift=-45pt] {};
\node[xscale=0.4,yscale=0.25,state,shape=ellipse] (m3) [below of=k1, xshift=45pt] {};
\node[xscale=0.4,yscale=0.25,state,shape=ellipse] (k2) [right of=m3, xshift=100pt] {};

\node[xscale=0.4,yscale=0.25,state,shape=ellipse] (n0) [below of=m0, yshift=-200pt,xshift=-60pt] {};
\node[xscale=0.4,yscale=0.25,state,shape=ellipse] (n1) [below of=m0, yshift=-200pt,xshift=-0pt] {};
\node[xscale=0.4,yscale=0.25,state,shape=ellipse] (n2) [below of=m1, yshift=-200pt,xshift=-30pt] {};
\node[xscale=0.4,yscale=0.25,state,shape=ellipse] (n3) [below of=m1, yshift=-200pt,xshift=30pt] {};
\node[xscale=0.4,yscale=0.25,state,shape=ellipse] (n5) [below of=m2, yshift=-200pt,xshift=0pt] {};
\node[xscale=0.4,yscale=0.25,state,shape=ellipse] (n6) [below of=m2, yshift=-200pt,xshift=60pt] {};
\node[xscale=0.4,yscale=0.25,state,shape=ellipse] (n7) [below of=m3, yshift=-200pt,xshift=45pt] {};
\node[xscale=0.4,yscale=0.25,state,shape=ellipse] (n8) [below of=m3, yshift=-200pt,xshift=105pt] {};
\node[xscale=0.4,yscale=0.25,state,shape=ellipse] (n9) [below of=m3, yshift=-200pt,xshift=165pt] {};
\node[xscale=0.4,yscale=0.25,state,shape=ellipse] (n10) [below of=m3,yshift=-200pt, xshift=280pt] {};
\node[yscale=0.25,draw=none] (p0) [above of=n0,xshift=15pt,yshift=-10pt] {};
\node[yscale=0.25,draw=none] (p1) [above of=n10,xshift=-5pt,yshift=-10pt] {};
\node[yscale=0.25,draw=none] (p2) [above of=n0, yshift=120pt,xshift=15pt] {};
\node[yscale=0.25,draw=none] (p3) [above of=n10, yshift=120pt,xshift=-5pt] {};
\node[yscale=0.25,draw=none,shape=ellipse, inner xsep=5pt] (q) [right of=p3] {};
\node[yscale=0.25,draw=none,shape=ellipse,inner xsep=5pt] (r) [above of=q,yshift=80pt] {};

\path[draw] 
 (l0) edge (k0.north)
 (l0) edge (k1.north)
 (l0) edge (k2.north)

 (k0) edge (m0.north)
 (k0) edge (m1.north)
 (k1) edge (m2.north)
 (k1) edge (m3.north)
 (k2) edge[dotted] (n10)
 (m0) edge[dotted] (n0)
 (m0) edge[dotted] (n1)
 (m1) edge[dotted] (n2)
 (m1) edge[dotted] (n3)
 (m2) edge[dotted] (n5)
 (m2) edge[dotted] (n6)
 (m3) edge[dotted] (n7)
 (m3) edge[dotted] (n8)
 (m3) edge[dotted] (n9)

;
           
\fill[fill=white] (p0.west) -- (p1.east) -- (p3.east) -- (p2.west) -- (p0.west);

\draw [decorate, decoration={snake,amplitude=.4mm, segment length=6mm, post length=0mm}] (p0.west) -- (p1.east);
\draw [decorate, decoration={snake,amplitude=.4mm, segment length=6mm, post length=0mm}] (p2.west) -- (p3.east);

\draw (l0|-p2) node [below,yshift=-3pt] {\scriptsize Software (FSMs)}
      (l0|-p0) node [above,yshift=3pt] {\scriptsize Devices};

\path[->] 
(q.north west) edge node[above,sloped,yshift=5pt] {\scriptsize Status/Alarms} (r.north west) ;
\path[->]
(r.north east) edge node[above,sloped,yshift=5pt] {\scriptsize Commands} (q.north east)
;

\end{tikzpicture}

\caption{Architecture of the real-time monitoring and control system of the CMS
experiment, running at the LHC.}

\label{fig:SMI}
\end{figure}

The control, configuration, readout and monitoring of hardware devices and
the detector status, in particular various kinds of environment variables
such as temperature, humidity, high voltage, and low voltage, are carried
out by the Detector Control System (DCS).  The control software of the
CMS detector is implemented with the Siemens commercial Supervision,
Control And Data Acquisition (SCADA) package PVSS-II and the CERN Joint
Controls Project (JCOP) framework \cite{Holme:907906}. The architecture
of the control software for all four big LHC experiments is based on the
SMI++ framework \cite{franek2005smi++,gaspar1998smi++}. Under the SMI++
framework, the real world is viewed as a collection of objects behaving
as finite state machines (FSMs). These FSMs are described using the State
Manager Language (SML).  A characteristic of the used architecture is the
regularity and relatively low complexity of the individual FSMs and device
drivers that together constitute the control software; the main source of
complexity is in the cooperation of these FSMs. Cooperation is strictly
hierarchical, consisting of several layers; see Figure~\ref{fig:SMI}
for a schematic overview. The FSMs are organised in a tree structure
where every node has one parent and zero or more children, except for
the top node, which has no parent. Nodes communicate by sending commands
to their children and state updates to their parents, so commands are
refined and propagated down the hierarchy and status updates are sent
upwards. Hardware devices are typically found only at the bottom-most
layer.

The FSM system in the CMS experiment contains over 30,000 nodes.
On average, each FSM contains 5 logical states.  Based on our early
experiments with some subsystems, we believe that
10$^\text{30,000}$ states is a very conservative estimate of the
size of the state space for the full control system.
The sheer size of the system significantly
contributes to its complexity.  Complicating factors in understanding the
behaviour of the system are the diversity in the development philosophies
in subgroups responsible for controlling their own subdetectors,
and the huge amount of parameters to be monitored.  In view of this
complexity, it is currently impossible to trace the root cause of
problems when unexpected behaviours manifest themselves. A single
badly designed FSM may be sufficient to lead to a livelock, resulting
in non-responsive hardware devices, potentially ruining expensive and
difficult experiments. Considering the scientific importance of these
experiments, this justifies the use of rigorous methods for understanding
and analysing the system.

Our contributions are twofold. First, we have formalised SML by mapping
its language constructs onto constructs in the process algebraic language
mCRL2~\cite{mcrl2ref}.  Second, based on our understanding of the
semantics of SML, we have identified properties that can be verified for
FSMs in isolation, and for which we have developed dedicated verification
tooling.

Using the ASF+SDF meta-environment~\cite{brand-meta}, we have developed
a prototype translation implementing our mapping of SML to mCRL2.
This allowed us to quickly assess the correctness of the translation
through simulation and visualisation of FSMs in isolation, and by means of
formal verification of small subsystems of the control software, using the
mCRL2 toolset.  The feedback obtained by the verification and simulation
enabled us to further improve the transformation. The use of the ASF+SDF
meta-environment allowed us to repeat this cycle in quick successions,
and, at the same time, maintain a formal description of the translation.
Although the ASF+SDF Meta Environment development was discontinued in
2010, we chose it over similar products as ATL because we were already
familiar with it and because its syntax-driven, functional approach
results in very clear translation rules.

Our dedicated verification tools allow the developers at CERN to quickly
perform behavioural sanity checks on their design, and use the feedback
of the tools to further improve on their designs in case of any problems.
Results using these tools so far are favourable: with only a fraction
of the total number of FSMs inspected so far, several problems have
surfaced and have been fixed.

\paragraph{Outline} We give a cursory overview of the core of the SML
language in Section~\ref{sec:cmsfsm}. The mCRL2 semantics of this core
are then explained in Section~\ref{sec:formalisation},
and we briefly elaborate on the methodology we used for obtaining this
semantics.  Our dedicated verification tools for SML, together with
some of the results obtained so far, are described in further detail
in Section~\ref{sec:dedicated_tooling}. We summarise our findings and
suggestions in Section~\ref{sec:conclusions}.

\section{The State Manager Language}
\label{sec:cmsfsm}

The finite state machines used in the CMS experiment are described in the
State Manager Language (SML) \cite{franek2005smi++,gaspar1998smi++}. We
present the syntax and the suggested meaning of the core of the language
using snapshots of a running example; we revisit this example in our
formalisation in Section~\ref{sec:formalisation}. Note that in reality,
SML is larger than presented here, but the control system is made up
largely of FSMs employing these core constructs only.

\begin{codelisting}
\begin{verbatim}
class: $FWPART_$TOP$RPC_Chamber_CLASS
    state: OFF
        when ( ( $ANY$FwCHILDREN in_state ERROR ) or
               ( $ANY$FwCHILDREN in_state TRIPPED ) )  move_to ERROR

        when ( $ANY$RPC_HV in_state {RAMPING_UP,
                                     RAMPING_DOWN} ) move_to RAMPING
        when ( ( $ALL$RPC_LV in_state ON ) and
               ( $ALL$RPC_HV in_state STANDBY ) )  move_to STANDBY

        when ( ( $ALL$RPC_HV in_state ON ) and
               ( $ALL$RPC_LV in_state ON ) )  move_to ON

        when ( ( $ALL$FwCHILDREN in_state ON ) and
               ( $ALL$RPC_T in_state OK ) )  move_to ON

        action: STANDBY
            do STANDBY $ALL$RPC_HV
            do ON $ALL$RPC_LV
        action: OFF
            do OFF $ALL$FwCHILDREN
        action: ON
            do ON $ALL$FwCHILDREN
\end{verbatim}
  \caption{Part of the definition of the \emph{Chamber} class in SML.}
  \label{code:chamber}
\end{codelisting}

Listing~\ref{code:chamber} shows part of the definition of a \emph{class}
in SML. Conceptually, this is the same kind of class known from
object-oriented programming: the class is defined once, but can be
instantiated many times. An instantiation is referred to as a
Finite State Machine.  A class consists of one or more \emph{state
clauses}; Listing~\ref{code:chamber} only shows the state clause for the
\texttt{OFF} state. Intuitively, a state clause describes how the FSM should
behave when it is in a particular state. Every state clause consists
of a list of \emph{when clauses} and a list of \emph{action clauses},
either of which may be empty.

A \whenclause has two parts: a \emph{guard} which is a Boolean expression
over the states of the children of the FSM and a \emph{referer} which
describes what should happen if the guard evaluates to true. The base
form of a guard is {\tt P in\_state S}, where \texttt{S} is the name of a state
(or a set of state names) and \texttt{P} is a \emph{child pattern}. A child
pattern consists of two parts: the first part is either \texttt{ANY} or
\texttt{ALL} and the second part is the name of a class or the literal
\texttt{FwCHILDREN}. The intended meaning is straightforward:

\begin{quote}
\texttt{\$ALL\$FwCHILDREN
  in\_state ON}
\end{quote}
means ``all children are in the \texttt{ON} state'', and:
\begin{quote}
\texttt{\$ANY\$RPC\_HV in\_state \{RAMPING\_UP,} \texttt{RAMPING\_DOWN\}}
\end{quote}
evaluates to true if ``some child of class \texttt{RPC\_HV} is either in state
\texttt{RAMPING\_UP} or state \texttt{RAMPING\_DOWN}''.

A referer is either of the form \texttt{move\_to S}, indicating that the
finite state machine changes its state to \texttt{S}, or of the form
\texttt{do A}, indicating that the action with name \texttt{A} should
be executed next. If the guards of more than one \whenclause evaluate
to true, the topmost enabled referer is executed. Whenever the FSM
moves to a new state, it executes the \whenclauses, starting from the
top \whenclause, to see if it should stay in this state (all guards are
false) or if it should go to another state (some guard is true). It is
therefore possible that a single \texttt{move\_to} referer or statement
(see below) triggers a series of state changes.

An \actionclause consists of a \emph{name} and a list of
\emph{statements}. When an FSM receives a command while in a state
\texttt{S}, it looks inside the state clause of state \texttt{S}
for an \actionclause with the same name as the command and if such an
\actionclause exists, it executes its statement list. If no such action
exists, the command is ignored. For example, if the \emph{Chamber}
finite state machine from Listing~\ref{code:chamber} is in state
\texttt{OFF} and it receives an \texttt{ON} command, it will execute
the last \actionclause.

The most commonly used statement is \texttt{do C P}, which means that
the command \texttt{C} is sent to all children which match the child pattern
\texttt{P}. After a command is sent, the child is marked \emph{busy}. When a
child sends its new state back, this \emph{busy} flag is removed. The
\texttt{do} statement is non-blocking, \ie, it does not wait for the
children to respond with their new state. The child pattern always
starts with \texttt{\$ALL\$} in this context. SML also provides
\texttt{if} and \texttt{move\_to} statements, as we illustrated in
Listing~\ref{code:statements}.

\begin{codelisting}
\begin{verbatim}
action: STANDBY
    do STANDBY $ALL$RPC_HV
    do ON $ALL$RPC_LV
    if ( $ALL$RPC_LV in_state ON ) then
        do ON $ALL$RPC_HV
        do ON $ALL$RPC_LV
        if ( $ALL$RPC_HV in_state ON ) then
            do ON $ALL$RPC_HV
            move_to ON
        endif
    else
        do STANDBY $ALL$RPC_LV
        do STANDBY $ALL$RPC_HV
        do STANDBY $ALL$FwCHILDREN
    endif
\end{verbatim}
  \caption{An example of a more complex \actionclause.}
  \label{code:statements}
\end{codelisting}

The \texttt{move\_to S} statement immediately stops execution of the
\actionclause and causes the FSM to move to the \texttt{S} state. The \texttt{if
G then S1 else S2 endif} statement blocks as long as there is a child,
referred to in \texttt{G}, that has a busy flag. If the guard \texttt{G}
evalutates to true, then \texttt{S1} is executed and otherwise \texttt{S2}
is executed. The else clause is optional.

\section{A Formal Semantics for SML} \label{sec:formalisation}

We use the process algebra mCRL2~\cite{mcrl2ref} to formalise the
semantics of programs written in SML.
\iffsen The formal translation of SML into mCRL2 
can be found in~\cite{HKW:11}.
\else
The formal translation of SML into mCRL2 can be found in the appendices.
\fi

 Our choice for mCRL2 is motivated
largely by the expressive power of the language, its rich data language
rooted in the theory of Abstract Data Types, its available tool support,
and our understanding of the advantages and disadvantages of mCRL2.
Before we address the translation of SML to mCRL2, we briefly
describe the mCRL2 language.

\subsection{A Brief Overview of mCRL2}

The mCRL2 language consists of two distinct parts: a \emph{data
language} for describing the data transformations and data types, and
a \emph{process language} for specifying system behaviours. For
a comprehensive language tutorial, we refer to \url{http://mcrl2.org}.

The data language, which is rooted in the theory of \emph{abstract data
types}, includes built-in definitions for many of the commonly used
data types, such as Booleans, Integers, Natural numbers, \emph{etc.},
and allows users to specify their own data sorts. In addition, container
sorts, such as \emph{lists}, \emph{sets} and \emph{bags} are available.

The process specification language of mCRL2 consists of only a small
number of basic operators and primitives. The language is inspired
by process algebras such as ACP~\cite{BBR:10}, and has both an
axiomatic and an operational semantics. 

A set of (parameterised) actions are used to model atomic, observable
events. Processes are constructed compositionally: 
the non-deterministic choice between processes \texttt{p}
and \texttt{q} is denoted \texttt{p+q}; their sequential composition
is denoted \texttt{p.q}, and their parallel composition is denoted
\texttt{p||q}. In addition, there are facilities to enforce
communication between different actions and abstracting from actions.

The main feature of the process language is that processes can depend
on data. For instance, \texttt{b->p<>q} denotes a conditional choice
between processes \texttt{p} and \texttt{q}: if \texttt{b} evaluates to
\emph{true}, it behaves as process \texttt{p}, and otherwise as process
\texttt{q}. In a similar vein, \texttt{sum d:D.p(d)} describes a (possibly
infinite) choice between processes \texttt{p} with different
values for variable \texttt{d}.

\subsection{From SML to mCRL2}
\label{sec:translation}

We next present our formalisation of SML in mCRL2.  Every SML class is
converted to an mCRL2 process definition; the behaviour of an FSM is then
described by the behaviour of a process instance.  Each FSM maintains a
state and a pointer to the code it is currently executing. In addition,
an FSM is embedded in a global tree-like configuration that identifies its
parent, and its children. In order to faithfully describe the behaviour
of an FSM, we therefore equip each mCRL2 process definition for a class
\texttt{X} with this information as follows:

\begin{verbatim}
proc X_CLASS(self: Id, parent: Id, s: State, chs: Children,
             phase: Phase, aArgs: ActPhaseArgs)
\end{verbatim}

Parameter \texttt{self} represents a unique identifier for a process instance,
and \texttt{parent} is the identifier of \texttt{self}'s parent in the
tree. Parameter \texttt{s} is used to keep track of the state of the
FSM. The state information of \texttt{self}'s children is stored in 
\texttt{chs} of sort \texttt{Children}, which is a list
of sort \texttt{Child}, a structured sort:

\begin{verbatim}
Children = List(Child);
Child = struct child(id:Id, state:State, ptype:PType, busy:Bool);
\end{verbatim}

\noindent The above structured sort \texttt{Child} can be thought of
as a named tuple; \texttt{id} represents the unique identifier of a
child, \texttt{state} is the state that this child sent to \texttt{X}
in its last state-update message, \texttt{ptype} maintains the FSM class of
this child, and \texttt{busy} is the flag that indicates that the child
is still processing the last command \texttt{X} sent to it. 
This flag is set after sending a message to the child, and reset when it
responds with its new state. Whenever \texttt{X} receives a state-update
message from one of its children, the \texttt{chs} structure is updated
accordingly. This structure is used to evaluate the \whenclauses and to
determine to which processes commands have to be sent.

The \texttt{phase} parameter has value \texttt{WhenPhase} if the FSM
is executing the \whenclauses and \texttt{ActionPhase} otherwise;
\texttt{Phase} is a simple structured sort containing these two
values. The phases will be explained in detail in the following
section. Finally, \texttt{aArgs} is a structure that contains information
we only need in the \actionphase. It is defined as follows:

\begin{verbatim}
ActPhaseArgs = struct actArgs(cq: CommandQueue, nrf: IdList,
                              pc: Int, rsc: Bool)
\end{verbatim}

\noindent
We forego a discussion of the \texttt{nrf} and \texttt{rsc} parameters,
which are solely used during an intialisation phase. The command
queue \texttt{cq} contains messages
that are to be sent to an FSM's children. Specifically, when executing
a \texttt{do C P} statement, we add a pair with the child's id and the
command \texttt{C} to \texttt{cq}, for every child matching the child
pattern \texttt{P}. The command queue is subsequently emptied by 
sending the messages stored in \texttt{cq}.

\subsubsection{Phases}

\begin{figure}

  \centering
  \begin{tikzpicture}[->,>=stealth',node distance=60pt]
  \tikzstyle{every state}=[minimum width=60pt, inner sep=2pt, shape=circle]
  \node[state,draw=none] (wp) {\whenphase};
  \node[state,draw=none] (ap) [right of=wp,xshift=30pt] {\actionphase};
  \node[state,shape=rectangle] (waiting) [below of=ap, xshift=60pt,yshift=20pt]
     {\begin{tabular}{c}waiting for\\ command or\\state-update\end{tabular}};
  \node[state,shape=rectangle] (executing) [below of=waiting] 
     {\begin{tabular}{c}executing\\ statements\end{tabular}};
  \node[state,shape=rectangle] (emptying) [below of=executing]
     {\begin{tabular}{c}emptying\\ command\\ queue \end{tabular}};
  \node[state,shape=rectangle] (evaluating) [left of=executing, xshift=-175pt]
     {\begin{tabular}{c}evaluating\\ when clauses\end{tabular}};
  \node[draw=none] (b) [below of=wp, xshift=45pt,yshift=-130pt] {};

  \draw 
  (evaluating.45) |- node[xshift=40pt, above] {all guards are false} 
(waiting.200);
  \draw
  (waiting.165) -| node[xshift=70pt, above] {receive state-update} (evaluating.135);
  \draw
  (waiting) -- node[right] {received command} (executing);
  \draw
  (executing) -- node[right] {\begin{tabular}{l}executed\\ last statement \end{tabular}} (emptying);
  \draw
  (emptying) -| node[xshift=70pt, above] {command queue is empty} (evaluating);
  \draw[dotted,>=,thick] (b) -- ++(0,200pt);

  \end{tikzpicture}
  \caption{Overview of the \whenphase and the \actionphase.}
  \label{fig:phases_flow_diagram}
\end{figure}
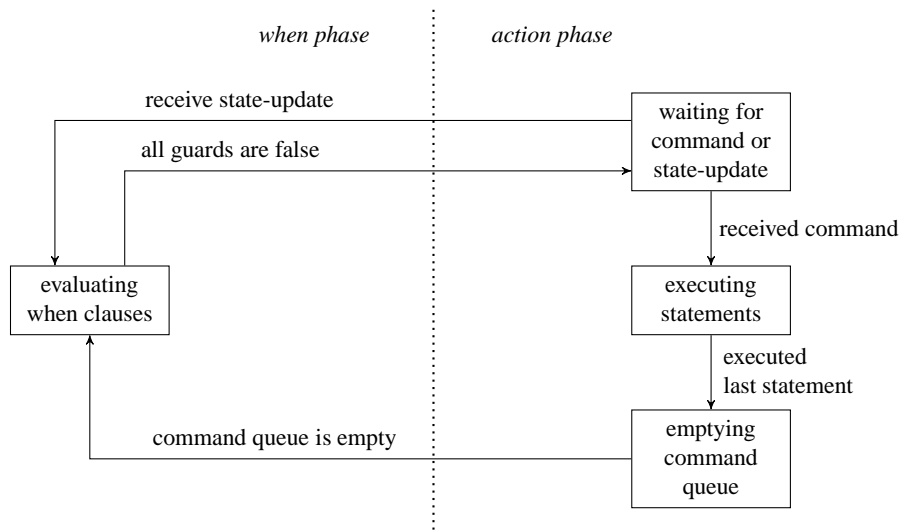

During the \whenphase, a process executes \whenclauses until it reaches
a state in which none of the guards evaluate to true. It then moves to
the \actionphase. In the \actionphase, a process can receive a command
from its parent or a state-update message from one of its children. This
process is illustrated in Figure~\ref{fig:phases_flow_diagram}. After
handling the command or message, it returns to the \whenphase. 

Translating the \whenphase turns out to be rather straightforward: for
each state a process term consisting of nested if-then-else statements is
introduced, formalised by mCRL2 expressions of the form \texttt{b->p<>q}
(if \texttt{b}, then act as process \texttt{p}, otherwise as \texttt{q}).
Each if-clause represents exactly one \whenclause.  The else-clause
of the last \whenclause sends a state-update message (represented by
the mCRL2 action \texttt{send\_state}) with the current state to the
parent of this FSM and moves to the \actionphase. An example is given
in Translation~\ref{translation:whens}.

\begin{trans}
  \begin{tabular}{p{0.32\textwidth} | p{0.02\textwidth} p{0.65\textwidth}}
    \textbf{SML} & & \textbf{mCRL2}\\
\begin{verbatim}
state: OFF
  when G1 move_to S1
  ...
  when Gn move_to Sn
\end{verbatim}
  & &
\begin{verbatim}
instate_OFF(s) && isWhenPhase(phase) -> (
  translation_of_G1 ->
    move_state(self,S1).
    X_CLASS(self,parent,S1,chs,phase,aArgs) <>
  ...
  translation_of_Gn ->
     move_state(self,Sn).
     X_CLASS(self,parent,Sn,chs,phase,aArgs) <>
  send_state(self,parent,s).
  move_phase(self,ActionPhase).
  X_CLASS(self,parent,s,chs,ActionPhase,
          reset(aArgs)))
\end{verbatim}
  \end{tabular}
  \caption{Simplified translation of the \whenclauses of a state OFF.
Note that \texttt{p.q} describes the process \texttt{p} that, upon
successful termination, continues to behave as process \texttt{q}.}
  \label{translation:whens}
\end{trans}

The \texttt{move\_state} action indicates that the process changes
its state. The \texttt{send\_state} action communicates with the
\texttt{receive\_state} action to a \texttt{comm\_state} action,
representing the communication of the new state to the parent. Note
that the state is sent only if none of the guards are true. Upon sending
the new state, the process changes to the \actionphase, signalled by a
\texttt{move\_phase} action.

Modelling the \actionphase is more involved as we need to add some terms
for initialisation and sending messages. We will focus on the translation
of the \actionclauses and the code which handles state-update messages.

SML allows for an arbitrary number of statements and an arbitrary
number of (nested) if-statements in every \actionclause. 
We uniquely identify the translation of every statement with
an integer label. After executing a statement, the \texttt{pc(aArgs)}
program counter is set to the label of the statement which should be
executed next. There are two special cases here:

\begin{itemize}
\item Label 0, the clause selector. When entering the \actionphase,
  the program counter is set to 0. Upon receiving a command, the clause
  selector sets the program counter to the label of the first statement
  of the \actionclause that should handle the command.

\item Label -1, end of action. After executing an action, the program
  counter is set to -1, signalling that the command queue must be emptied
  and the process must change to the \whenphase.

\end{itemize}

\noindent An example is given in
Translation~\ref{translation:actions}. The \texttt{receive\_command}
action models the reception of a command that was sent by the FSM's
parent.  Such a command is ignored if no
\actionclause handles it. In the example, observe that both after ignoring
a command and after completing the execution of the \texttt{STANDBY}
action handler, the program counter is set to -1. A process term not
shown here then empties the command queue by issueing a sequence of
\texttt{send\_command} actions, and subsequently returns to
the \whenphase. Note that these \texttt{send\_command} actions and
\texttt{receive\_command} actions are meant to synchronise, resulting
in a \texttt{comm\_command} action. This is enforced at a higher level
in the specification.

\begin{trans}
  \begin{tabular}{p{0.32\textwidth} | p{0.02\textwidth} p{0.65\textwidth}}
    \textbf{SML} & & \textbf{mCRL2}\\
\begin{verbatim}
state: OFF
  action: STANDBY
    do STANDBY $ALL$Y
    do ON $ALL$Z
  action: OFF
    do OFF $ALL$Y
  action: ON
    do ON $ALL$Y
\end{verbatim}
  & &
\begin{verbatim}
instate_OFF(s) && isActPhase(phase) -> (
  pc(aArgs) == 0 ->
    sum c:Command.(
      receive_command(parent,self,c).
      isC_STANDBY(c) ->
        X_CLASS(self,parent,s,chs,phase,
                update_pc(aArgs,1)) <>
      isC_OFF(c) ->
        X_CLASS(self,parent,s,chs,phase,
                update_pc(aArgs,3)) <>
      isC_ON(c) ->
        X_CLASS(self,parent,s,chs,phase,
                update_pc(aArgs,4)) <>
      send_state(self,parent,s).
      ignored_command(self,c).
      X_CLASS(self,parent,s,chs,phase,
              update_pc(aArgs,-1))) +

  pc(aArgs) == 1 ->
    RPC_Chamber_CLASS(self,parent,s,chs,phase,
      add_HV_STANDBY_commands(
        update_pc(aArgs,2))) +

  pc(aArgs) == 2 ->
    RPC_Chamber_CLASS(self,parent,s,chs,phase,
      add_LV_ON_commands(
        update_pc(aArgs,-1)) + ...
\end{verbatim}
  \end{tabular}
  \caption{Simplified translation of the \actionclauses of a state OFF.}
  \label{translation:actions}
\end{trans}

Since a \texttt{do} statement is asynchronous, the children can send
their state-update at any time during the \actionphase. This is dealt
with as follows.  Suppose a state-update message is received. If this
precedes the reception of a command in this \actionphase, we simply
process the state-update and move to the \whenphase. If we are in the
middle of executing an \actionclause, we process the state-update,
but do not move to the \whenclause.

\subsection{Validating the Formalisation of SML}

The challenge in formalising SML is in correctly interpreting
its language constructs. We combined two strategies for assessing and
improving the correctness of our semantics: informal discussions with the
development team of the language and applying formal analysis techniques
on sample FSMs taken from the control software.

The discussions with the SML development team were used to solidify
our initial understanding of SML and its main constructs. Based on these
discussions, we manually translated several FSMs into mCRL2, and validated
the resulting processes manually using the available simulation and
visualisation tools of mCRL2. This revealed a few minor issues with
our understanding of the semantics of SML, alongside many issues that
could be traced back to sloppiness in applying the translation from SML
to mCRL2 manually.

In response to the latter problem, we eliminated the need for manually
translating FSMs to mCRL2. To this end, we utilised the ASF+SDF meta-environment (see~\cite{brand-meta,klint-meta}) to rapidly prototype an
automatic translator that, ultimately, came to implement the translation
scheme we described in the previous section.  The \emph{Syntax Definition
Formalism} (SDF) was used to describe the syntax of both SML and mCRL2,
whereas the \emph{Algebraic Specification Formalism} (ASF) was used
to express the term rewrite rules that are needed to do the actual
translation.  Apart from the gains in speed and the consistency in
applying the transformations that were brought about by the automation,
the automation also served the purpose of formalising the semantics
of SML.

The final details of our semantics were tested by analysing relatively
well-understood subsystems of the control software in mCRL2. We
briefly discuss our findings using a partly simplified subsystem,
colloquially known as the \emph{Wheel}, see Figure~\ref{fig:wheel_system}.
The Wheel subsystem is a component of the Resistive Plate Chamber (RPC)
subdetector of the CMS experiment. It belongs to the barrel region
of the RPC subdetector. Each Wheel subsystem contains 12 sectors,
each sector is equipped with 4 muon stations which are made of Drift
Tube chambers.  We forego a detailed formal discussion of this subsystem
(for details, we refer to~\cite{PP:08}), but only address our analysis
of this subsystem using formal analyses techniques, and the impact this
had on our understanding of the semantics and the transformation. It is
important to keep in mind that the analysis was conducted primarily to
assess the quality of our translation, the correctness of the subsystem
being only secondary.

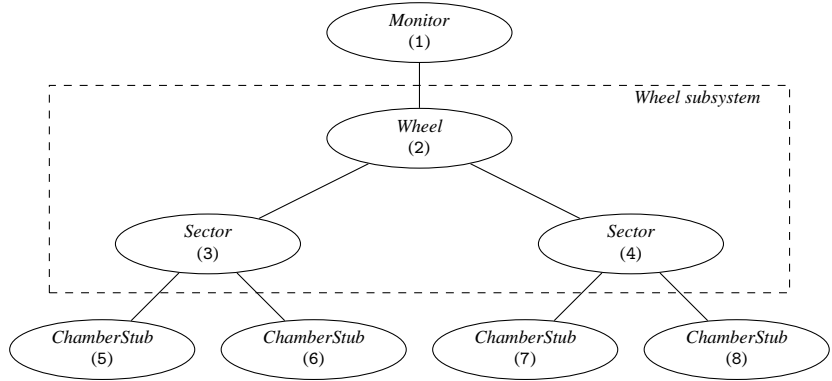
\begin{figure}
  \centering
  \begin{tikzpicture}[>=,node distance=40pt]
  \tikzstyle{every state}=[minimum width=70pt, inner sep=0pt]

  \node[state,ellipse] (top) {
  \scriptsize\begin{tabular}{c} \emph{Monitor}\\ (\texttt{1})\end{tabular}};
  \node[state,ellipse] (wheel) [below of=top] 
  {\scriptsize\begin{tabular}{c} \emph{Wheel}\\ (\texttt{2})\end{tabular}};
  \node[state,ellipse] (sector3) [below of=wheel, xshift=-80pt] 
  {\scriptsize\begin{tabular}{c} \emph{Sector}\\ (\texttt{3})\end{tabular}};
  \node[state,ellipse] (sector4) [below of=wheel,xshift=80pt] 
  {\scriptsize\begin{tabular}{c} \emph{Sector}\\ (\texttt{4})\end{tabular}};
  \node[state,ellipse] (chamber5) [below of=sector3,xshift=-40pt] 
  {\scriptsize\begin{tabular}{c} \emph{ChamberStub}\\ (\texttt{5})\end{tabular}};
  \node[state,ellipse] (chamber6) [below of=sector3,xshift=40pt] 
  {\scriptsize\begin{tabular}{c} \emph{ChamberStub}\\ (\texttt{6})\end{tabular}};
  \node[state,ellipse] (chamber7) [below of=sector4,xshift=-40pt] 
  {\scriptsize\begin{tabular}{c} \emph{ChamberStub}\\ (\texttt{7})\end{tabular}};
  \node[state,ellipse] (chamber8) [below of=sector4,xshift=40pt] 
  {\scriptsize\begin{tabular}{c} \emph{ChamberStub}\\ (\texttt{8})\end{tabular}};

  \path 
    (top) edge (wheel)
    (wheel) edge (sector3) edge (sector4)
    (sector3) edge (chamber5) edge (chamber6)
    (sector4) edge (chamber7) edge (chamber8);    
  \draw[dashed] (chamber5.north)[yshift=10pt,xshift=-20pt] rectangle (160pt,-30pt);
  \draw (wheel) node [xshift=105pt,yshift=15pt] {\scriptsize\emph{Wheel subsystem}};

  \end{tikzpicture}
  \caption{A schematic overview of our model of the \emph{Wheel}
    subsystem, and its used FSMs. The identifiers of the processes
    representing the FSMs are given between parentheses; these were
    used in our analyses.}

  \label{fig:wheel_system}
\end{figure}

The mCRL2 specification of the \emph{Wheel} subsystem was obtained
by combining the mCRL2 processes obtained by running our prototype
implementation on each involved FSM. Generating the state space of the
\emph{Wheel} subsystem takes roughly one minute using the symbolic state
space generation tools offered by the LTSmin tools~\cite{BlomPW10}. This
toolset can be integrated in the mCRL2 toolset. For the discussed
configuration, the state space is still of modest proportions, measuring
slightly less than 5 million states and 24 million transitions. Varying
the amount of children of class \emph{Sector} causes a dramatic
growth of the state space.  Using 3 instead of 2 children of class
\emph{Sector} yields roughly 800 million states; using 4 children of
class \emph{Sector}, leads to 120 billion states,  and requires
half a day.

Apart from repeating the simulations and visualisations, at this stage
we also applied \emph{model checking} to systematically probe the
translation. Together with the development team of the \emph{Wheel}
subsystem, a few basic requirements were formalised in the first-order
modal $\mu$-calculus~\cite{GW:05}, see Table~\ref{tab:requirements}.
The first-order modal $\mu$-calculus is the default requirement
specification language in the mCRL2 toolset.

\begin{table}[h]
\caption{Basic requirements for the \emph{Wheel} subsystem;
\texttt{i:Id} denotes an identifier of an FSM;
\texttt{i\_c:Id} denotes a child of FSM \texttt{i};
\texttt{c:Command} denotes a command;\texttt{c2s(c)} denotes the
state with the homonymous command name, \eg, \texttt{c2s(ON) = ON}.
}
\label{tab:requirements}

\centering
\begin{tabular}{c}
\hline\\
\parbox{.8\textwidth}
{
\begin{enumerate}
\item Absence of deadlock:

\begin{quote}
{\small \tt nu X. [true]X \&\& <true>true} 
\end{quote}

\item Absence of intermediate states in the \whenphase:

\begin{quote}
{\small \tt nu X. [true]X \&\& 

\qquad [exists s:State. move\_state(i,s)](nu Y. 

\qquad\qquad [(!move\_phase(i,ActionPhase))]Y

\qquad~\&\& [exists s:State. move\_state(i,s)]false)}
\end{quote}

\item Responsiveness:

\begin{quote}
{\small \tt nu X. [true]X \&\& 

\qquad [comm\_command(i,i\_c,c)](mu Y.

\qquad \qquad <true>true \&\& [!comm\_state(i\_c,i,c2s(c))]Y)}
\end{quote}

\item Progress:

\begin{quote}
{\small \tt nu X. [true]X \&\&

\qquad mu Y. <exists s:State. move\_state(i,s)>true || 

\qquad\qquad (<true>true \&\& [true]Y)
}
\end{quote}

\end{enumerate}
}\\
\hline
\end{tabular}

\end{table}

The studied subsystem was considered to satisfy all stated properties.
While smoothing out details in the translation of SML to mCRL2, the
deadlock-freedom property was violated every now and then, indicating
issues with our interpretation of SML. These were mostly concerned with
the semantics of the blocking and non-blocking constructs of SML, and
the complex constructs used to model the message passing between FSMs
and their children.

The absence of intermediate states in the \whenphase was violated only
once in our verification efforts. A more detailed scrutiny of the
run revealed a problem in our translation, which was subsequently fixed.

The third requirement, stating the inevitability of a state change by a
child once such a state change has been commissioned,
failed to hold.  The violation is caused by the overriding of commands by
subsequent commands that are issued immediately.  Discussions with the
development teams revealed that the violations are real, \ie, they are
within the range of real behaviour, suggesting that our formalisation
was adequate. The property was modified to ignore the spurious runs,
resulting in the following property:

\begin{quote}
{\small \tt nu X. [true]X \&\&

\qquad [comm\_command(i,i\_c,c)](mu Y. <true>true \&\& 

\qquad\qquad [!(comm\_state(i\_c,i,c2s(c)) ||

\qquad\qquad \ \ \ exists c':Command. comm\_command(i,i\_c,c'))]Y)
}
\end{quote}

The final requirement also failed to hold.  The violation is similar
spirited to the violation of the third requirement, and, again found to
comply to reality. The weakened requirement that was subsequently agreed
upon expresses the attainability of some state change:

\begin{quote}
{\small \tt nu X. [true]X \&\&

\qquad mu Y. <exists s:State. move\_state(i,s)>true || <true>Y
}
\end{quote}
Neither visual inspection of the state space using 2D and 3D visualisation
tools, nor simulation using the mCRL2 simulators revealed any further
incongruences in our final formalisation of SML, sketched in the previous
section.

\section{Dedicated Tooling for Verification}
\label{sec:dedicated_tooling}

Some desired properties, such as the absence of loops within the
\whenphase, can be checked by analysing an FSM in isolation, using the
transformation to mCRL2.  However, the verifications using the modal
$\mu$-calculus currently require too much overhead to serve as a basis
for lightweight tooling that can be integrated in the SML development
environment. 

In an attempt to improve on this situation, we explored the
possibilities of using \emph{Bounded Model Checking} (BMC)
\cite{biere1999symbolic,biere2003bounded}. The basic idea of BMC is to
check for a counterexample in bounded runs.  If no bugs are found using
the current bound, then the bound is increased until either a bug is
found, the problem becomes intractable, or some pre-determined upper
bound is reached upon which the verification is complete.  The BMC
problem can be efficiently reduced to a propositional satisfiability
problem, and can therefore be solved by SAT methods. SAT procedures do
not necessarily suffer from the space explosion problem, and a modern
SAT solver can handle formulas with hundreds of thousands of variables
or more, see \eg~\cite{biere2003bounded}.

We have applied BMC techniques for the detection of \texttt{move\_to}
loops and the detection of unreachable states and trap states. As
an example of a \texttt{move\_to} loop, consider the excerpt of the
\texttt{ECALfw\_CoolingDee} FSM class in Listing~\ref{code:endlessloop}, which
our tool found to contain issues. If an
instance of \texttt{ECALfw\_CoolingDee} has one child in state
\texttt{ERROR} and one in state \texttt{NO\_CONNECTION}, it will loop
indefinitely between these two states. Once this happens, an entire
subsystem may enter a livelock and become unresponsive.

\begin{codelisting}
\begin{verbatim}
state: ERROR
  when ( $ANY$FwCHILDREN in_state NO_CONNECTION ) move_to NO_CONNECTION
  when ( $ALL$FwCHILDREN in_state OK ) move_to OK

state: NO_CONNECTION
  when ( $ALL$FwCHILDREN in_state OK ) move_to OK
  when ( $ANY$FwCHILDREN in_state ERROR ) move_to ERROR
\end{verbatim}
  \caption{An excerpt from the \texttt{ECALfw\_CoolingDee} FSM that exhibits a
    loop within the \whenphase.}
  \label{code:endlessloop}
\end{codelisting}
We first convert this problem into a graph problem as follows. Let $\FF$
be an FSM and $\MM$ be a Kripke structure. A state in $\MM$ corresponds
to the combined state of $\FF$ and its children, \eg, if $\FF$ is in
state \texttt{ON} and has two children which are in state \texttt{OFF},
then the corresponding state in $\MM$ is $(\texttt{ON}, \texttt{OFF},
\texttt{OFF})$. There is a transition between two states $s_1$ and $s_2$
in $\MM$ if and only if $s_1$ can do a \texttt{move\_to} action to $s_2$
in $\FF$. Moreover, every state in $\MM$ is an initial state.
It thus suffices to inspect $\MM$ instead of $\FF$, as stated by
the following lemma:
\begin{lemma}
  \label{lem:loopifkripke}
  $\FF$ contains a loop of \texttt{move\_to} actions if and only if $\MM$ contains
  a loop.
\end{lemma}

\newcommand{\instate}{\mathit{in\_state}}
\newcommand{\nextstate}{\mathit{next\_state}}

We next translate the problem of detecting a loop in $\MM$ into a SAT
problem. First, we consider executions of length $k$; afterwards, we show that
we can statically choose $k$ such that we can find every loop.

Let the predicate $\instate$ be defined as follows: $\instate(s, p, i)$
holds if and only if the process with identifier $p$ is in state $s$ after
$i$ steps. We assign the identifier zero to the FSM under consideration
and the numbers $1,2,3,\dots$ to its children. The resulting formula will
have three components: the \emph{state constraints}, the \emph{transition
relation} and the \emph{loop condition}.

Using the state constraints, we ensure the FSM to always be in exactly
one state. Moreover, the states of the children should not change during
the execution of the \whenphase, per the semantics in the previous
section. This is straigthforwardly expressed as a boolean formula on
the $\instate$ predicate.

Next, we encode the transition relation: the relation between $\instate(s,
0, i)$ and $\instate(s', 0, i+1)$ for every $i$. In other words: the
\texttt{move\_to} steps the parent process is allowed to take. This
involves converting the \whenclauses for each state of the parent
FSM, taking care the semantics as outlined in the previous section is
reflected. The last ingredient is the loop condition: if $\instate(s, 0,
0)$ holds, then $\instate(s, 0, i)$ must hold for some $i>1$, indicating
that the parent returned to the state in which it started.

The final SAT formula is obtained by taking the conjunction of the state
constraints, the transition relation and the loop condition.
It is not hard to see that if this formula is satisfiable, then there is a loop in
\MM and hence in \FF. 
It is more difficult to show that if there is a loop, then the formula is
satisfiable. Let $n$ be the total number of states of the FSM and let $n_t$ be
the total number of states of each child class $t$. We then have the following
result:
\begin{theorem}
  \label{thm:find_all_loops}
  All possible loops in \FF can be found by considering paths of length at most
  $n$ in an FSM configuration \FF having $n_t$ children for each child class
  $t$.
\end{theorem}
\begin{proof}[sketch]
  Since \FF only has $n$ states, the longest possible loop also contains $n$
  states. Since every state in $\MM$ is an initial state, every possible loop can by
  found by doing $n$ steps from an initial state.

  It remains to show that all loops can be found by considering
  a configuration with $n_t$ children for each child class $t$. This
  follows from the fact that SML guards are restricted to check for
  \emph{any} or \emph{all} children in a particular state.\qed

\end{proof}
A second desirable behavioural property of an FSM is that all states
should remain reachable during the execution of an FSM.  While we
can again easily encode this property into the modal $\mu$-calculus,
we use a more direct approach to detect violations of this property
by constructing a graph that captures all potential state changes.
For this, we determine whether there is a configuration of children
such that \FF can execute a \texttt{move\_to} action from a state $s$
to a state $s'$. Doing so for all pairs $(s,s')$ of states of \FF yields
a graph encoding all possible state changes of \FF.

Computing the strongly connected components (SCCs) of the thusly
obtained graph gives sufficient information to pinpoint violations to
the reachability property: the presence of more than a single SCC means
that one cannot move back and forth these SCCs (by definition of an SCC),
and, therefore, their states.  Note that this is an under-approximation
of all errors that can potentially exist, as the actual reachability
dynamically depends on the configuration of the children of an
FSM. Still, as the state change graph of the \texttt{ESfw\_Endcap}
FSM class in Figure~\ref{fig:move_to_graph} illustrates, issues can be
found in production FSMs: the \texttt{OFF} state can never be reached
from any of the other states. Using the graphs generated by our tools,
such issues are quickly explained and located.
\begin{figure}[h!]
  \centering

\begin{tikzpicture}[>=stealth', node distance=45pt]

\node (on) [draw, shape=ellipse] {\tiny\texttt{ON}};
\node (hvramping) [draw, shape=ellipse, below of=on] {\tiny\texttt{HV\_RAMPING}};
\node (partlyon) [draw, shape=ellipse, below left of=hvramping,xshift=-70pt] {\tiny\texttt{PARTLY\_ON}};
\node (lvonhvoff) [draw, shape=ellipse, below right of=hvramping,xshift=70pt] {\tiny\texttt{LV\_ON\_HV\_OFF}};
\node (offlocked) at (partlyon-|hvramping) [draw, shape=ellipse] {\tiny\texttt{OFF\_LOCKED}};
\node (error) at (hvramping-|partlyon) [draw, shape=ellipse] {\tiny\texttt{ERROR}};
\node (off) at (on-|lvonhvoff) [draw,shape=ellipse,yshift=40pt] {\tiny\texttt{OFF}};

\path[<->] 
 (hvramping) edge (on)
 (hvramping) edge (partlyon)
 (hvramping) edge (lvonhvoff)
 (partlyon) edge (on)
 (partlyon) edge (error)
 (partlyon) edge (offlocked)
 (partlyon) edge[bend right] (lvonhvoff)
 (on) edge (lvonhvoff)
;
\path[->,dotted] 
 (hvramping) edge (error)
 (on) edge[bend left] (offlocked)
 (lvonhvoff) edge (offlocked)
 (hvramping) edge[bend right] (offlocked)
 (on) edge (error)
 (lvonhvoff) edge (error);

\path[->,dotted] 
 (off) edge (on)
 (off) edge[bend right] (error)
 (off) edge (partlyon)
 (off) edge (hvramping)
 (off) edge (lvonhvoff)
;

\node[shape=rectangle,draw=black,inner xsep=15pt, inner ysep=10pt,dashed] (offbox) at (off) {};

\node[shape=rectangle,draw=black,inner xsep=130pt, inner ysep=67pt,dashed] (hvrampingbox) at (hvramping) {};

\end{tikzpicture}

  \caption{The state change graph for the \texttt{ESfw\_Endcap} FSM class. 
    The solid lines are bidirectional; the dotted lines are unidirectional 
    state changes. The SCCs are indicated by the dashed frames.}

  \label{fig:move_to_graph}
\end{figure}
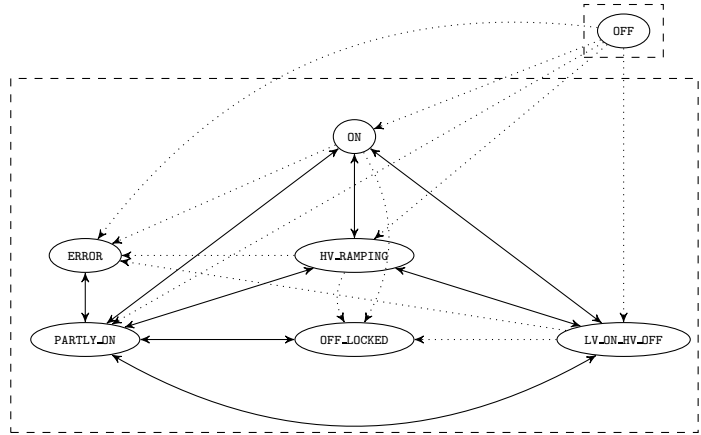

\paragraph{Results}

The results using our dedicated tools for performing these behavioural
sanity checks on isolated FSMs are very satisfactory: of the several
hundreds of FSM classes contained in the control system, we so far
analysed 40 FSM classes and found 6 to contain issues.  In 4 of these,
we found logical errors that could give rise to livelocks in the system
due to the presence of loops in the \whenphase; an example thereof is
given in Listing~\ref{code:endlessloop}. Somewhat unexpectedly, all loops
were found to involve two states. Note that the size of the average
FSM class (in general more than 100 lines of SML code, and at least
two children) means that even short loops such as the ones identified
so far remain unnoticed and are hard to pinpoint. The remaining two
FSM classes were found to violate the required reachability of states,
see \eg~Figure~\ref{fig:move_to_graph}. The speed at which the errors can
be found (generally requiring less than a second) means that 
the sanity checks could easily be incorporated in the design cycle of
the FSMs.

\section{Conclusion}
\label{sec:conclusions}

We discussed and studied the State Machine Language (SML) that is
currently used for programming the control software of the CMS experiment
running at the Large Hadron Collider. To fully understand the language,
we formalised it using the process algebraic language mCRL2. The quality
of our formalisation was assessed using a combination of simulation and
visualisation of the behaviour of FSMs in isolation and formally verifying
small subsystems using model checking. To facilitate, among others,
the assessment, the translation of SML to mCRL2 was implemented using
the ASF+SDF meta-environment.  Based on our understanding
of the semantics of SML, we have built dedicated tools for performing
sanity checks on isolated FSMs. Using these tools we found several
issues in the control system.  These tools have been well-received by
the engineers at CERN, and are considered for inclusion in the
development environment.

Our formalisation of SML opens up the possibility of verifying
realistically large subsystems of the control system; clearly, it will be
one of the most challenging verification problems currently available.
In our analysis of the \emph{Wheel} subsystem, we have only used a
modest set of tools for manipulating the state space; symmetry reduction,
partial order reduction, parallel exploration techniques, abstractions
and abstract interpretation were not considered at this point. It remains
to be investigated how such techniques fare on this problem.

\paragraph{Acknowledgments.}
We thank Giel Oerlemans, Dennis Schunselaar and Frank Staals from the
Eindhoven University of Technology for their contribution to a
first draft of the ASF+SDF translation.
We also
thank Frank Glege and Robert Gomez-Reino Garrido from the CERN CMS DAQ
group for their support and advice, and Clara Gaspar for discussions on
SML. Jaco van de Pol is thanked for his help with the LTSmin toolset.

\bibliographystyle{plain}

\iffsen

\else

\fi

\iffsen
\end{document}

\else

\appendix
\cleardoublepage

\section{ASF and SDF files}
\subsection{midtools.sdf}
\tiny
\begin{listing}[5]{1}

module midtools
imports
  basic/Whitespace
  basic/Comments
  basic/Booleans
  basic/Integers

exports
sorts
  MId
  MIds

lexical restrictions
  MId -/- [a-zA-Z0-9\_\']

lexical syntax
  [a-zA-Z\_] ([a-zA-Z0-9\_\'])* -> MId

context-free syntax
  {MId ","}+ -> MIds

  concat(MId, MId)                                              -> MId

  contains(MId, MId*)                                           -> Boolean

  removeDuplicates(MId*)                                        -> MId*

  remove(MId, MId*)                                             -> MId*

  intersect(MId*, MId*)                                         -> MId*

  empty(MId*)                                                   -> Boolean

  length(MId*)                                                  -> Integer

hiddens
variables
  "$mid"[0-9]*                      -> MId
  "$mid+"[0-9]*                     -> MId+
  "$mid*"[0-9]*                     -> MId*
  "$i"                              -> Integer

lexical variables
  "#midHead"[0-9]*                  -> [a-zA-Z\_]
  "#midTailChar"[0-9]*              -> ([a-zA-Z0-9\_\'])
  "#midTail"[0-9]*                  -> ([a-zA-Z0-9\_\'])*
\end{listing}

\subsection{midtools.asf}
\tiny
\begin{listing}[5]{1}
equations

[concat-1]
concat(mid(#midHead1 #midTail1), mid(#midHead2 #midTail2)) = mid(#midHead1 #midTail1 #midHead2 #midTail2)

[contains-empty]
contains($mid, ) = false

[contains-match]
contains($mid, $mid $mid*) = true

[contains-nomatch]
$mid1 != $mid2
===>
contains($mid1, $mid2 $mid*) = contains($mid1, $mid*)

[removeDuplicates-empty]
removeDuplicates() =

[removeDuplicates-many-nonunique]
contains($mid, $mid*) == true
===>
removeDuplicates($mid $mid*) = removeDuplicates($mid*)

[removeDuplicates-many-unique]
contains($mid, $mid*) == false
===>
removeDuplicates($mid $mid*) = $mid removeDuplicates($mid*)

[remove-empty]
remove($mid, ) =

[remove-many-nomatch]
$mid1 != $mid2
===>
remove($mid1, $mid2 $mid*) = $mid2 remove($mid1, $mid*)

[remove-many-match]
remove($mid,  $mid $mid*) = remove($mid, $mid*)

[intersect-empty]
intersect(, $mid*2) =

[intersect-many-match]
contains($mid, $mid*2) == true
===>
intersect($mid $mid*1, $mid*2) = $mid intersect($mid*1, $mid*2)

[intersect-many-nomatch]
contains($mid, $mid*2) == false
===>
intersect($mid $mid*1, $mid*2) = intersect($mid*1, $mid*2)

[empty-true]
empty() = true

[empty-false]
empty($mid+) = false

[length-empty]
length() = 0

[length-many]
$i := length($mid*)
===>
length($mid $mid*) = $i + 1
\end{listing}

\subsection{cfsm.sdf}
\tiny
\begin{listing}[5]{1}

module cfsm
imports
  basic/Comments
  basic/Whitespace
  midtools

exports
context-free start-symbols
  FSMSpecification

sorts
  Identifier
  FSMSpecification
  FSMClass
  FSMStateClause
  FSMWhenClause
  FSMReferer
  FSMActionClause
  FSMStatement
  FSMParameter

  FSMExpression

  FSMChildrenSpec
  FSMChildrenAny
  FSMChildrenAll
  FSMChildrenAnySpecific
  FSMChildrenAnyFwChildren
  FSMChildrenAllSpecific
  FSMChildrenAllFwChildren

  FSMClassName
  FSMStateName
  FSMStateNameSpec
  FSMActionName

lexical syntax
  "!" "=" ~[\n]* [\n] -> LAYOUT
  "/associated" ~[\n]* [\n] -> LAYOUT
  "/ASSOCIATED" ~[\n]* [\n] -> LAYOUT
  [A-zA-Z0-9]* -> Identifier

context-free syntax
  FSMClass+                                                  -> FSMSpecification

  "class: $FWPART_$TOP$" FSMClassName FSMStateClause+        -> FSMClass
  "state:" FSMStateName FSMWhenClause* FSMActionClause*      -> FSMStateClause
  "when" "(" FSMExpression ")" FSMReferer+                   -> FSMWhenClause
  "action:" FSMActionName FSMStatement*                      -> FSMActionClause

   Identifier                                           -> FSMParameter
  "move_to" FSMStateName                                -> FSMStatement
  "do" FSMActionName FSMChildrenSpec                    -> FSMStatement
  "do" FSMActionName "(" FSMParameter "=" FSMParameter ")" FSMChildrenSpec -> FSMStatement
  "if" "(" FSMExpression ")" "then" FSMStatement+ ("else" FSMStatement+)? "endif" -> FSMStatement

  "move_to" FSMStateName                                -> FSMReferer
  "do" FSMActionName                                    -> FSMReferer

  FSMChildrenSpec "in_state"     FSMStateNameSpec         -> FSMExpression
  FSMChildrenSpec "not_in_state" FSMStateNameSpec         -> FSMExpression
  "not" "(" FSMExpression ")"                             -> FSMExpression
  "not" "(" FSMExpression ")" "and" "(" FSMExpression ")" -> FSMExpression

  "(" FSMExpression ")"                                 -> FSMExpression
  FSMExpression "and" FSMExpression                     -> FSMExpression {left}
  FSMExpression "or"  FSMExpression                     -> FSMExpression {left}

  FSMStateName                                          -> FSMStateNameSpec
  "{" {FSMStateName ","}* "}"                           -> FSMStateNameSpec

  "(" FSMChildrenSpec ")"  -> FSMChildrenSpec
  FSMChildrenAny           -> FSMChildrenSpec
  FSMChildrenAll           -> FSMChildrenSpec

  FSMChildrenAnySpecific   -> FSMChildrenAny
  FSMChildrenAnyFwChildren -> FSMChildrenAny
  FSMChildrenAllSpecific   -> FSMChildrenAll
  FSMChildrenAllFwChildren -> FSMChildrenAll

  "$ANY$" FSMClassName     -> FSMChildrenAnySpecific
  "$ANY$FwCHILDREN"        -> FSMChildrenAnyFwChildren
  "$ANY$FwCHILDREN"        -> FSMChildrenAnySpecific {reject}

  "$ALL$" FSMClassName     -> FSMChildrenAllSpecific
  "$ALL$FwCHILDREN"        -> FSMChildrenAllFwChildren
  "$ALL$FwCHILDREN"        -> FSMChildrenAllSpecific {reject}

  MId -> FSMClassName
  MId -> FSMStateName
  MId -> FSMActionName
\end{listing}

\subsection{mcrlt.sdf}
\tiny
\begin{listing}[5]{1}

module mcrlt
imports
  basic/Whitespace
  basic/Comments
  basic/Integers
  midtools

exports
sorts

  SortExpr
  Domain
  SortSpec
  SortDecl
  ConstrDecl
  ProjDecl
  ProjDecls

  IdDecl
  IdsDecl
  OpSpec
  OpDecl

  EqnSpec
  EqnDecl

  DataExpr
  DataExprs
  BagEnumElt
  BagEnumElts

  MAId
  MAIdSet
  CommExpr
  CommExprSet
  RenExpr
  RenExprSet

  ProcExpr

  ActDecl
  ActSpec

  ProcDecl
  ProcSpec
  Init

  MCRL2Spec

lexical restrictions
  MId -/- [a-zA-Z0-9\_\']

context-free start-symbols
  MCRL2Spec

lexical syntax
  [\r\t\n\ ]       -> LAYOUT
  "


context-free syntax
  "sort"    -> MId {reject}
  "cons"    -> MId {reject}
  "map"     -> MId {reject}
  "var"     -> MId {reject}
  "eqn"     -> MId {reject}
  "act"     -> MId {reject}
  "proc"    -> MId {reject}
  "init"    -> MId {reject}
  "delta"   -> MId {reject}
  "tau"     -> MId {reject}
  "sum"     -> MId {reject}
  "block"   -> MId {reject}
  "allow"   -> MId {reject}
  "hide"    -> MId {reject}
  "rename"  -> MId {reject}
  "comm"    -> MId {reject}
  "struct"  -> MId {reject}
  "Bool"    -> MId {reject}
  "Pos"     -> MId {reject}
  "Nat"     -> MId {reject}
  "Int"     -> MId {reject}
  "Real"    -> MId {reject}
  "List"    -> MId {reject}
  "Set"     -> MId {reject}
  "Bag"     -> MId {reject}
  "true"    -> MId {reject}
  "false"   -> MId {reject}
  "whr"     -> MId {reject}
  "end"     -> MId {reject}
  "lambda"  -> MId {reject}
  "forall"  -> MId {reject}
  "exists"  -> MId {reject}
  "div"     -> MId {reject}
  "mod"     -> MId {reject}
  "in"      -> MId {reject}

  "Bool"                                 -> SortExpr
  "Pos" | "Nat" | "Int" | "Real"         -> SortExpr
  "List" "(" SortExpr ")"                -> SortExpr
  "Set" "(" SortExpr ")"                 -> SortExpr
  "Bag" "(" SortExpr ")"                 -> SortExpr
  MId                                    -> SortExpr
  "(" SortExpr ")"                       -> SortExpr
  Domain "->" SortExpr                   -> SortExpr

  {SortExpr "#"}+                         -> Domain
  "sort" SortDecl+                        -> SortSpec
  MIds ";"                                -> SortDecl
  MIds "=" SortExpr ";"                   -> SortDecl
  MIds "=" "struct" {ConstrDecl "|"}+ ";" -> SortDecl

  MId ("(" ProjDecls ")")? ("?" MId)       -> ConstrDecl
  (MId ":")? Domain                        -> ProjDecl
  {ProjDecl ","}+                          -> ProjDecls

  MId ":" SortExpr          -> IdDecl
  MIds ":" SortExpr ";"?    -> IdsDecl
  ("cons" | "map") OpDecl+  -> OpSpec
  IdsDecl ";"               -> OpDecl

  "eqn" EqnDecl+                          -> EqnSpec
  "var" IdsDecl+ "eqn" EqnDecl+           -> EqnSpec
  DataExpr "=" DataExpr ";"               -> EqnDecl
  DataExpr "->" DataExpr "=" DataExpr ";" -> EqnDecl

  MId | Integer | "true" | "false" | "[]" | "{}" -> DataExpr
  "[" DataExprs "]"                            -> DataExpr
  "{" DataExprs "}"                            -> DataExpr
  "{" BagEnumElts "}"                          -> DataExpr
  "{" IdDecl "|" DataExpr "}"                  -> DataExpr
  "(" DataExpr ")"                             -> DataExpr
  ("!" | "#") DataExpr                         -> DataExpr
  ("forall" | "exists") IdDecl "." DataExpr    -> DataExpr
  "lambda" IdDecl "." DataExpr                 -> DataExpr
  DataExpr "whr" DataExprs "end"               -> DataExpr

  {DataExpr ","}+                              -> DataExprs
  DataExpr ":" DataExpr                        -> BagEnumElt
  {BagEnumElt ","}*                            -> BagEnumElts

  {MId "|"}+               -> MAId
  "{" {MAId ","}* "}"      -> MAIdSet
  MAId ("->" MId)?         -> CommExpr
  "{" {CommExpr ","}* "}"  -> CommExprSet
  MId "->" MId             -> RenExpr
  "{" {RenExpr ","}* "}"   -> RenExprSet

  MId                                                        -> ProcExpr
  MId "(" DataExprs ")"                                      -> ProcExpr
  "delta"                                                    -> ProcExpr
  "tau"                                                      -> ProcExpr
  ("block" | "allow" | "hide") "(" MAIdSet "," ProcExpr ")"  -> ProcExpr
  "rename" "(" RenExprSet "," ProcExpr ")"                   -> ProcExpr
  "comm" "(" CommExprSet "," ProcExpr ")"                    -> ProcExpr
  "(" ProcExpr ")"                                           -> ProcExpr

  MIds (":" Domain)? ";"  -> ActDecl
  "act" ActDecl+         -> ActSpec

  MId "=" ProcExpr ";"                         -> ProcDecl
  MId "(" {IdsDecl ","}+ ")" "=" ProcExpr ";"  -> ProcDecl
  "proc" ProcDecl+                             -> ProcSpec
  "init" ProcExpr ";"                          -> Init

  SortSpec* OpSpec* EqnSpec* ActSpec* ProcSpec* Init* -> MCRL2Spec

context-free priorities
  ProcExpr "@" ProcExpr                 -> ProcExpr         >
  "sum" {IdDecl ","}+ "." ProcExpr      -> ProcExpr         >
  ProcExpr "." ProcExpr                 -> ProcExpr {right} >
  ProcExpr "<<" ProcExpr                -> ProcExpr {left}  >
  {
    ProcExpr "||"  ProcExpr             -> ProcExpr {right}
    ProcExpr "|"   ProcExpr             -> ProcExpr {right}
    ProcExpr "||_" ProcExpr             -> ProcExpr {right}
  }                                                         >
  DataExpr "->" ProcExpr "<>" ProcExpr  -> ProcExpr         >
  DataExpr "->" ProcExpr                -> ProcExpr         >
  ProcExpr "+" ProcExpr                 -> ProcExpr {right} >

  DataExpr "(" DataExprs ")"                                 -> DataExpr        >
  DataExpr ("|>" | "<|") DataExpr                            -> DataExpr {left} >
  DataExpr "++" DataExpr                                     -> DataExpr {left} >
  DataExpr ("." | "*" | "div") DataExpr                      -> DataExpr {left} >
  DataExpr ("+" | "-") DataExpr                              -> DataExpr {left} >
  DataExpr ("mod"  | "in") DataExpr                          -> DataExpr {left} >

  DataExpr ("==" | "!=" | "<" | ">" | "<=" | ">=") DataExpr  -> DataExpr {left} >
  DataExpr ("&&" | "||" | "=>") DataExpr                     -> DataExpr {left}
\end{listing}

\subsection{cfsm2mcrl2.sdf}
\tiny
\begin{listing}[5]{1}

module cfsm2mcrl2

imports cfsm
imports mcrlt
imports genericclauses
imports midtools
imports basic/Integers
imports basic/BoolCon

exports

context-free start-symbols
SortDecl+ ProcExpr

context-free syntax

  cfsm2mcrl2(FSMClass+) -> ProcSpec+

  cfsm2mcrl2bm(FSMClass+) -> ProcSpec+

  fsmGenerateSorts(FSMClass+) -> SortDecl+

  mcrl2GetPTypes(ProcSpec+) -> SortDecl
hiddens

sorts
  ProcName ActionClauseTuple UniqueProcName PC

context-free syntax

  MId                                                           -> ProcName
  Integer                                                       -> PC
  <FSMActionName, PC, DataExpr, ProcExpr>                       -> ActionClauseTuple

  fsmClasses2Mcrl2Procs(FSMClass+, BoolCon) -> ProcSpec+

  fsmClass2Mcrl2Proc(FSMClass, BoolCon)     -> ProcSpec
  fsmClassName2ProcName(MId, BoolCon)       -> MId

  convertStates(FSMStateClause*, ProcName, BoolCon, MId*)       -> ProcExpr
  convertState(FSMStateClause, ProcName, BoolCon, MId*)         -> ProcExpr


  convertWhenClauses(FSMWhenClause*, ProcName, MId, ActionClauseTuple*) -> ProcExpr
  convertReferer(FSMReferer, ProcName, MId, ActionClauseTuple*)         -> ProcExpr
  convertExpr(FSMExpression)                                            -> DataExpr
  convertChildrenSpec(FSMChildrenSpec)                                  -> DataExpr
  convertStateNameSpec(FSMStateNameSpec)                                -> DataExpr

  combineActionClauseComponents(ActionClauseTuple*, ProcName, FSMStateName) -> ProcExpr

  gatherComponentsFromActionClauses(FSMActionClause*, ProcName, MId, PC) -> <ActionClauseTuple*, PC>

  getActionClauseTupleForActionName(ActionClauseTuple*, FSMActionName) -> ActionClauseTuple

  constructClauseSelectors(ActionClauseTuple*, ProcName) -> ProcExpr

  convertStatements(FSMStatement*, ProcName, PC, PC, PC)  -> <ProcExpr,PC>
  convertStatement(FSMStatement, ProcName, PC, PC, PC)    -> <ProcExpr,PC>

  insertIfBlockingWaiter(ProcName, PC)                    -> ProcExpr

  inAnyState(MId*)                                                    -> DataExpr
  createObedientCommandAcceptor(FSMActionClause*, ProcName, MId*)     -> ProcExpr

  isStateCheck(MId)      -> DataExpr
  isStateCheck(MId, MId) -> DataExpr
  isCommandCheck(MId)    -> DataExpr
  "isStateCheck"         -> DataExpr {reject}
  "isCommandCheck"       -> DataExpr {reject}

  toMcrlIsFunction(MId) -> MId

  toMcrlStateName(MId) -> MId
  "toMcrlStateName"    -> DataExpr {reject}
  "toMcrlStateName"    -> MId {reject}

  toMcrlCmdName(MId)   -> MId



  convertClassNamesToTypeConstrDecl(MId*)   -> {ConstrDecl "|"}+
  convertStateNamesToStateConstrDecl(MId*)  -> {ConstrDecl "|"}+
  convertActionNamesToCmdConstrDecl(MId*)   -> {ConstrDecl "|"}+

  collectClasses(FSMSpecification,MId*)       -> MId* {traversal(accu,top-down,continue)}
  collectClasses(FSMClass,MId*)               -> MId* {traversal(accu,top-down,continue)}
  collectClasses(FSMChildrenAnySpecific,MId*) -> MId* {traversal(accu,top-down,continue)}
  collectClasses(FSMChildrenAllSpecific,MId*) -> MId* {traversal(accu,top-down,continue)}

  collectStates(FSMSpecification, MId*)     -> MId* {traversal(accu,top-down,continue)}
  collectStates(FSMStateClause+, MId*)      -> MId* {traversal(accu,top-down,continue)}
  collectStates(FSMStateClause, MId*)       -> MId* {traversal(accu,top-down,continue)}
  collectStates(FSMWhenClause*, MId*)       -> MId* {traversal(accu,top-down,continue)}
  collectStates(FSMWhenClause, MId*)        -> MId* {traversal(accu,top-down,continue)}
  collectStates(FSMStateNameSpec, MId*)     -> MId* {traversal(accu,top-down,continue)}

  collectCommands(FSMSpecification,MId*)    -> MId* {traversal(accu,top-down,continue)}
  collectCommands(FSMActionClause,MId*)     -> MId* {traversal(accu,top-down,continue)}

  collectActionNames(FSMSpecification, MId*) -> MId* {traversal(accu,top-down,continue)}
  collectActionNames(FSMActionClause*, MId*) -> MId* {traversal(accu,top-down,continue)}

  addToSet(MId,MId*)                        -> MId*

  mcrl2PTypesFromProcSpecs(ProcSpec+) -> {ConstrDecl "|"}+
  mcrl2PTypesFromProcDecls(ProcDecl+) -> {ConstrDecl "|"}+

variables
  "$mid"[0-9]*                      -> MId
  "$mid+"[0-9]*                     -> MId+
  "$mid*"[0-9]*                     -> MId*
  "$mids"[0-9]*                     -> MIds

  "$fsmSpec"[0-9]*                  -> FSMSpecification
  "$fsmClass"[0-9]*                 -> FSMClass
  "$fsmClass+"[0-9]*                -> FSMClass+
  "$fsmState"[0-9]*                 -> FSMStateClause
  "$fsmState+"[0-9]*                -> FSMStateClause+
  "$fsmWhenClause"[0-9]*            -> FSMWhenClause
  "$fsmWhenClause+"[0-9]*           -> FSMWhenClause+
  "$fsmWhenClause*"[0-9]*           -> FSMWhenClause*
  "$fsmActionClause"[0-9]*          -> FSMActionClause
  "$fsmActionClause*"[0-9]*         -> FSMActionClause*
  "$fsmExpr"[0-9]*                  -> FSMExpression
  "$fsmExpr*"[0-9]*                 -> FSMExpression*
  "$fsmStatement"[0-9]*             -> FSMStatement
  "$fsmStatement+"[0-9]*            -> FSMStatement+
  "$fsmStatement*"[0-9]*            -> FSMStatement*
  "$fsmReferer"[0-9]*               -> FSMReferer

  "$fsmClassName"[0-9]*             -> MId

  "$fsmChildrenSpec"[0-9]*          -> FSMChildrenSpec
  "$fsmStateNameSpec"[0-9]*         -> FSMStateNameSpec
  "$fsmStateNameSpecs*"[0-9]*       -> {FSMStateName ","}*
  "$fsmChildrenAnySpecific"[0-9]*   -> FSMChildrenAnySpecific
  "$fsmChildrenAllSpecific"[0-9]*   -> FSMChildrenAllSpecific
  "$fsmChildrenAnyFwChildren"[0-9]* -> FSMChildrenAnyFwChildren
  "$fsmChildrenAllFwChildren"[0-9]* -> FSMChildrenAllFwChildren
  "$fsmActionName"[0-9]*            -> MId 

  "$fsmStateName"[0-9]*             -> MId
  "$fsmCurrentState"[0-9]*          -> MId
  "$fsmNewState"[0-9]*              -> MId
  "$fsmStateNames"[0-9]*            -> {FSMStateName ","}+
  "$mcrl2CurrentState"[0-9]*        -> MId
  "$mcrl2NewState"[0-9]*            -> MId

  "$actionClauseTuple"[0-9]*        -> ActionClauseTuple
  "$actionClauseTuple+"[0-9]*       -> ActionClauseTuple+
  "$actionClauseTuple*"[0-9]*       -> ActionClauseTuple*

  "$procDecl"[0-9]*                 -> ProcDecl
  "$procDecl+"[0-9]*                -> ProcDecl+
  "$procDecl*"[0-9]*                -> ProcDecl*
  "$dataExpr"[0-9]*                 -> DataExpr
  "$dataExprs"[0-9]*                -> {DataExpr ","}+
  "$procExpr"[0-9]*                 -> ProcExpr
  "$procSpec"[0-9]*                 -> ProcSpec
  "$procSpec+"[0-9]*                -> ProcSpec+
  "$procSpec*"[0-9]*                -> ProcSpec*
  "$mcrl2Command"                   -> MId
  "$mcrlActionCondition"            -> DataExpr
  "$mcrlActionEffect"               -> ProcExpr

  "$procName"[0-9]*                 -> MId

  "$pc"[0-9]*                       -> Integer
  "$start_pc"[0-9]*                 -> Integer
  "$avail_pc"[0-9]*                 -> Integer
  "$jump_pc"[0-9]*                  -> Integer
  "$then_pc"[0-9]*                  -> Integer
  "$else_pc"[0-9]*                  -> Integer
  "$end_pc"[0-9]*                   -> Integer

  "$idsDecls"                       -> {IdsDecl ","}+

  "$b"                              -> BoolCon

lexical variables
  "#midHead"[0-9]*                  -> [a-zA-Z\_]
  "#midTailChar"[0-9]*              -> ([a-zA-Z0-9\_\'])
  "#midTail"[0-9]*                  -> ([a-zA-Z0-9\_\'])*
\end{listing}

\subsection{cfsm2mcrl2.asf}
\tiny
\begin{listing}[5]{1}
equations

[cfsm2mcrl2-1]
cfsm2mcrl2($fsmClass+) = fsmClasses2Mcrl2Procs($fsmClass+, false)

[cfsm2mcrl2bm-1]
cfsm2mcrl2bm($fsmClass+) = fsmClasses2Mcrl2Procs($fsmClass+, true)

[convertSpec-single]
fsmClasses2Mcrl2Procs($fsmClass, $b) = fsmClass2Mcrl2Proc($fsmClass, $b)

[convertSpec-multi]
fsmClasses2Mcrl2Procs($fsmClass $fsmClass+, $b) =
    fsmClass2Mcrl2Proc($fsmClass, $b) fsmClasses2Mcrl2Procs($fsmClass+, $b)


[fsmClassName2ProcName-bm]
fsmClassName2ProcName($mid, true) = $mid 

[fsmClassName2ProcName-nobm]
fsmClassName2ProcName($mid, false) = $mid

[fsmClass2Mcrl2Proc-1]
$procName := fsmClassName2ProcName($fsmClassName, $b),
$procExpr := convertStates($fsmState+, $procName, $b, collectStates($fsmState+, )),
$procSpec := proc $procName(self: Id, parent: Id, s: State, chs: Children, phase: Phase, aArgs: ActPhaseArgs) =
    $procExpr +
    insertGenericClauses($procName);
===>
fsmClass2Mcrl2Proc(class: $FWPART_$TOP$ $fsmClassName $fsmState+, $b) = $procSpec

[convertStates-1-element]
convertStates($fsmState, $procName, $b, $mid*) = convertState($fsmState, $procName, $b, $mid*)

[convertStates-list]
convertStates($fsmState $fsmState+, $procName, $b, $mid*) =
    convertState($fsmState, $procName, $b, $mid*) +
    convertStates($fsmState+, $procName, $b, $mid*)

[convertState-nobm]
<$actionClauseTuple*, $pc> := gatherComponentsFromActionClauses($fsmActionClause*, $procName, $fsmCurrentState, 1)
===>
convertState(state: $fsmCurrentState $fsmWhenClause* $fsmActionClause*, $procName, false, $mid*) =
    (


    convertWhenClauses($fsmWhenClause*, $procName, $fsmCurrentState, $actionClauseTuple*) +







    ((isStateCheck($fsmCurrentState)) && (isActPhase(phase)) && (!(initialized(chs))) &&
     (pc(aArgs) == 0) && (nrf(aArgs) == [])) ->
        start_initialization(self).
        $procName(self, parent, s, chs, phase,
                  actArgs([], children_to_ids(chs), 0, rsc(aArgs))) <>



    ((initialized(chs)) ->
    (
    (((isStateCheck($fsmCurrentState)) && (isActPhase(phase)) && (cq(aArgs) == []) && (pc(aArgs) == 0)) ->
        sum c:Command.(
           rc(parent, self, c).
           constructClauseSelectors($actionClauseTuple*, $procName)
    )) +

    combineActionClauseComponents($actionClauseTuple*, $procName, $fsmCurrentState)
    ))


    )

[convertState-bm]
$dataExpr := inAnyState(collectStates($fsmWhenClause*, )),
$procExpr := createObedientCommandAcceptor($fsmActionClause*, $procName, $mid*)
===>
convertState(state: $fsmCurrentState $fsmWhenClause* $fsmActionClause*, $procName, true, $mid*) =
    (


    ((isStateCheck($fsmCurrentState)) ->
    sum c: Command.
    (
    rc(parent, self, c).
    $procExpr
    )
    )


    )


[convertWhenClauses-empty]
$mcrl2CurrentState := toMcrlStateName($fsmCurrentState)
===>
convertWhenClauses(, $procName, $fsmCurrentState, $actionClauseTuple*) =
    (

    (((isStateCheck($fsmCurrentState)) && (isWhenPhase(phase))) ->
    ss(self, parent, s).
    move_phase(self, ActionPhase).
    $procName(self, parent, s, chs, ActionPhase, reset(aArgs)))

    )

[convertReferer-moveto]
$mcrl2NewState := toMcrlStateName($fsmNewState)
===>
convertReferer(move_to $fsmNewState, $procName, $mcrl2CurrentState, $actionClauseTuple*) =
    move_state(self, $mcrl2NewState).
    $procName(self, parent, $mcrl2NewState, chs, phase, aArgs)

[convertReferer-do]
<$fsmActionName, $start_pc, $mcrlActionCondition, $mcrlActionEffect> := getActionClauseTupleForActionName($actionClauseTuple*, $fsmActionName)
===>
convertReferer(do $fsmActionName, $procName, $mcrl2CurrentState, $actionClauseTuple*) =
    move_phase(self, ActionPhase).
    $mcrlActionEffect

[getActionClauseTupleForActionName-many-match]
<$fsmActionName, $start_pc, $mcrlActionCondition, $mcrlActionEffect> := $actionClauseTuple
===>
getActionClauseTupleForActionName($actionClauseTuple $actionClauseTuple*, $fsmActionName) =
    $actionClauseTuple

[getActionClauseTupleForActionName-many-nomatch]
<$fsmActionName2, $start_pc, $mcrlActionCondition, $mcrlActionEffect> := $actionClauseTuple,
$fsmActionName1 != $fsmActionName2
===>
getActionClauseTupleForActionName($actionClauseTuple $actionClauseTuple*, $fsmActionName1) =
    getActionClauseTupleForActionName($actionClauseTuple*, $fsmActionName1)

[convertWhenClauses-many]
$mcrl2CurrentState := toMcrlStateName($fsmCurrentState)
===>
convertWhenClauses(when ($fsmExpr) $fsmReferer $fsmWhenClause*, $procName, $fsmCurrentState, $actionClauseTuple*) =
    (
    ((isStateCheck($fsmCurrentState)) && (isWhenPhase(phase)) &&
     (convertExpr($fsmExpr))) ->
    convertReferer($fsmReferer, $procName, $mcrl2CurrentState, $actionClauseTuple*) <>

    (convertWhenClauses($fsmWhenClause*, $procName, $fsmCurrentState, $actionClauseTuple*))
    )


[gatherComponentsFromActionClauses-empty]
gatherComponentsFromActionClauses(, $procName, $fsmCurrentState, $pc) =
    <, $pc>

[gatherComponentsFromActionClauses-many]
$start_pc1 := $avail_pc1,
$avail_pc2 := $avail_pc1 + 1,
$mcrl2Command := toMcrlCmdName($fsmActionName),
<$procExpr, $avail_pc3> := convertStatements($fsmStatement*, $procName, $start_pc1, -1, $avail_pc2),
<$actionClauseTuple*, $avail_pc4> :=
    gatherComponentsFromActionClauses($fsmActionClause*, $procName, $fsmCurrentState, $avail_pc3)
===>
gatherComponentsFromActionClauses(action: $fsmActionName $fsmStatement* $fsmActionClause*,
                                  $procName, $fsmCurrentState, $avail_pc1) =
    <
    <
    $fsmActionName,

    $start_pc1,

    ((isStateCheck($fsmCurrentState)) && (isActPhase(phase)) && (cq(aArgs) == [])),

    ($procExpr)
    >
    $actionClauseTuple*, $avail_pc4>

[combineActionClauseComponents-empty]
combineActionClauseComponents(, $procName, $fsmCurrentState) = delta

[combineActionClauseComponents-many]
<$fsmActionName, $start_pc, $mcrlActionCondition, $mcrlActionEffect> := $actionClauseTuple
===>
combineActionClauseComponents($actionClauseTuple $actionClauseTuple*, $procName, $fsmCurrentState) =
    (

    ($mcrlActionCondition ->
    $mcrlActionEffect) +


    (combineActionClauseComponents($actionClauseTuple*, $procName, $fsmCurrentState)))

[constructClauseSelectors-empty]
constructClauseSelectors(, $procName) =

    ss(self, parent, s).
    ignored_command(self, c).
    $procName(self, parent, s, chs, phase, update_pc(aArgs, -1))


[constructClauseSelectors-many]
constructClauseSelectors(<$fsmActionName, $start_pc, $mcrlActionCondition, $mcrlActionEffect> $actionClauseTuple*,
                         $procName) =
    (isCommandCheck($fsmActionName) -> $procName(self, parent, s, chs, phase, update_pc(aArgs, $start_pc)) <> (
           constructClauseSelectors($actionClauseTuple*, $procName)))

[createObedientCommandAcceptor-empty]
createObedientCommandAcceptor(, $procName, $mid*) = constructClauseSelectors(, $procName)

[createObedientCommandAcceptor-many-nomatch]
contains($fsmActionName, $mid*) == false
===>
createObedientCommandAcceptor(action: $fsmActionName $fsmStatement* $fsmActionClause*, $procName, $mid*) =
    createObedientCommandAcceptor($fsmActionClause*, $procName, $mid*)

[createObedientCommandAcceptor-many-match]
contains($fsmActionName, $mid*) == true,
$mcrl2NewState := toMcrlStateName($fsmActionName),
$mid := toMcrlCmdName($fsmActionName)
===>
createObedientCommandAcceptor(action: $fsmActionName $fsmStatement* $fsmActionClause*, $procName, $mid*) =
    ((c == $mid) ->
    ss(self, parent, $mcrl2NewState).
    move_state(self, $mcrl2NewState).
    $procName(self, parent, $mcrl2NewState, chs, WhenPhase, (reset(aArgs)))
    <>
    createObedientCommandAcceptor($fsmActionClause*, $procName, $mid*))


[convertStatements-empty]
convertStatements(, $procName, $start_pc, $jump_pc, $avail_pc) =
    <
    (
    ((pc(aArgs) == $start_pc) ->
         noop_statement(self).
         ($procName(self, parent, s, chs, phase, update_pc(aArgs, $jump_pc))))
    )
    , $avail_pc>

[convertStatements-single]
convertStatements($fsmStatement, $procName, $start_pc, $jump_pc, $avail_pc) =
    convertStatement($fsmStatement, $procName, $start_pc, $jump_pc, $avail_pc)

[convertStatements-multiple]
$start_pc2 := $avail_pc1,
$avail_pc2 := $avail_pc1 + 1,
<$procExpr1, $avail_pc3> := convertStatement($fsmStatement, $procName, $start_pc1, $start_pc2, $avail_pc2),
<$procExpr2, $avail_pc4> := convertStatements($fsmStatement+, $procName, $start_pc2, $jump_pc, $avail_pc3)
===>
convertStatements($fsmStatement $fsmStatement+, $procName, $start_pc1, $jump_pc, $avail_pc1) =
    <$procExpr1 +
     $procExpr2, $avail_pc4>

[convertStatement-do]
convertStatement(do $fsmActionName $fsmChildrenSpec, $procName, $start_pc, $jump_pc, $avail_pc) =
    <
    (
    ((pc(aArgs) == $start_pc) ->
         queue_messages(self).
         ($procName(self, parent, s, chs, phase,
                    actArgs(send_command(toMcrlCmdName($fsmActionName),
                                         convertChildrenSpec($fsmChildrenSpec)), [],  $jump_pc, rsc(aArgs)))))
    )
    , $avail_pc>

[convertStatement-moveto]
$mcrl2NewState := toMcrlStateName($fsmNewState)
===>
convertStatement(move_to $fsmNewState, $procName, $start_pc, $jump_pc, $avail_pc) =
    <
    (
    ((pc(aArgs) == $start_pc) ->
          (ss(self, parent, $mcrl2NewState).
           move_phase(self, WhenPhase).
           $procName(self, parent, $mcrl2NewState, chs, ActionPhase, reset(aArgs))))
    )
    , $avail_pc>

[insertIfBlockingWaiter-1]
insertIfBlockingWaiter($procName, $pc) =
sum s1:State.(
         rs(id(head(busy_children(chs))), self, s1).
         $procName(self, parent, s,
                   update_busy(id(head(busy_children(chs))),
                               false,
                               update_state(id(head(busy_children(chs))), s1, chs)),
                   phase, update_pc(aArgs, $pc)))

[convertStatement-ifthenend]
$start_pc2 := $avail_pc1,
$avail_pc2 := $avail_pc1 + 1,
<$procExpr1, $avail_pc3> :=
    convertStatements($fsmStatement+, $procName, $start_pc2, $jump_pc, $avail_pc2),
$procExpr2 :=
    (
    ((pc(aArgs) == $start_pc1) ->
    (
    (busy_children(chs) != []) ->
      (
      insertIfBlockingWaiter($procName, $start_pc1)
      )
    <>
      (
      ((convertExpr($fsmExpr)) ->
        enter_then_clause(self).
        $procName(self, parent, s, chs, phase, update_pc(aArgs, $start_pc2)) <>
        skip_then_clause(self).
        $procName(self, parent, s, chs, phase, update_pc(aArgs, $jump_pc)))
      )
      )) +

      (
       $procExpr1
      )

    )
===>
convertStatement(if ( $fsmExpr ) then $fsmStatement+ endif, $procName, $start_pc1, $jump_pc, $avail_pc1) =
    <
    $procExpr2,
    $avail_pc3
    >

[convertStatement-ifthenelseend]
$start_pc2 := $avail_pc1,
$start_pc3 := $avail_pc1 + 1,
$avail_pc2 := $avail_pc1 + 2,
<$procExpr1, $avail_pc3> :=
    convertStatements($fsmStatement+1, $procName, $start_pc2, $jump_pc, $avail_pc2),
<$procExpr2, $avail_pc4> :=
    convertStatements($fsmStatement+2, $procName, $start_pc3, $jump_pc, $avail_pc3),
$procExpr3 :=
    (
    ((pc(aArgs) == $start_pc1) ->
    (
    (busy_children(chs) != []) ->
      (
      insertIfBlockingWaiter($procName, $start_pc1)
      )
    <>
      (
      ((convertExpr($fsmExpr)) ->
        enter_then_clause(self).
        $procName(self, parent, s, chs, phase, update_pc(aArgs, $start_pc2)) <>
        enter_else_clause(self).
        $procName(self, parent, s, chs, phase, update_pc(aArgs, $start_pc3)))
      )
      )) +

      (
       $procExpr1
      ) +

      (
      $procExpr2
      )

    )
===>
convertStatement(if ( $fsmExpr ) then $fsmStatement+1 else $fsmStatement+2 endif, $procName, $start_pc1, $jump_pc, $avail_pc1) =
    <
    $procExpr3,
    $avail_pc4
    >


[convertExpr-and]
convertExpr($fsmExpr0 and $fsmExpr1) = convertExpr($fsmExpr0) && convertExpr($fsmExpr1)
[convertExpr-or]
convertExpr($fsmExpr0 or $fsmExpr1) = convertExpr($fsmExpr0) || convertExpr($fsmExpr1)


[convertExpr-allchildren]
convertExpr($fsmChildrenAllFwChildren in_state $fsmStateNameSpec) =
    all_in_state(convertChildrenSpec($fsmChildrenAllFwChildren), convertStateNameSpec($fsmStateNameSpec))
[convertExpr-anychildren]
convertExpr($fsmChildrenAnyFwChildren in_state $fsmStateNameSpec) =
    any_in_state(convertChildrenSpec($fsmChildrenAnyFwChildren), convertStateNameSpec($fsmStateNameSpec))

[convertExpr-allinstatespecific]
convertExpr($fsmChildrenAllSpecific in_state $fsmStateNameSpec) =
    all_in_state(convertChildrenSpec($fsmChildrenAllSpecific), convertStateNameSpec($fsmStateNameSpec))
[convertExpr-anyinstatespecific]
convertExpr($fsmChildrenAnySpecific in_state $fsmStateNameSpec) =
    any_in_state(convertChildrenSpec($fsmChildrenAnySpecific), convertStateNameSpec($fsmStateNameSpec))

[convertExpr-notallchildren]
convertExpr(not ($fsmChildrenAllFwChildren) in_state $fsmStateNameSpec) =
    !(all_in_state(convertChildrenSpec($fsmChildrenAllFwChildren), convertStateNameSpec($fsmStateNameSpec)))

[convertExpr-notallinstatespecific]
convertExpr(not ($fsmChildrenAllSpecific) in_state $fsmStateNameSpec) =
    !(all_in_state(convertChildrenSpec($fsmChildrenAllSpecific), convertStateNameSpec($fsmStateNameSpec)))

[convertExpr-allchildren-not]
convertExpr($fsmChildrenAllFwChildren not_in_state $fsmStateNameSpec) =
    !(any_in_state(convertChildrenSpec($fsmChildrenAllFwChildren), convertStateNameSpec($fsmStateNameSpec)))
[convertExpr-anychildren-not]
convertExpr($fsmChildrenAnyFwChildren not_in_state $fsmStateNameSpec) =
    !(all_in_state(convertChildrenSpec($fsmChildrenAnyFwChildren), convertStateNameSpec($fsmStateNameSpec)))
[convertExpr-allinstatespecific-not]
convertExpr($fsmChildrenAllSpecific not_in_state $fsmStateNameSpec) =
    !(any_in_state(convertChildrenSpec($fsmChildrenAllSpecific), convertStateNameSpec($fsmStateNameSpec)))
[convertExpr-anyinstatespecific-not]
convertExpr($fsmChildrenAnySpecific not_in_state $fsmStateNameSpec) =
    !(all_in_state(convertChildrenSpec($fsmChildrenAnySpecific), convertStateNameSpec($fsmStateNameSpec)))

[convertExpr-bracket]
convertExpr(($fsmExpr)) = (convertExpr($fsmExpr))

[convertChildrenSpec-alltype]
convertChildrenSpec($ALL$ $fsmClassName) = filter_children(chs, concat($fsmClassName, _CLASS))
[convertChildrenSpec-anytype]
convertChildrenSpec($ANY$ $fsmClassName) = filter_children(chs, concat($fsmClassName, _CLASS))
[convertChildren-all]
convertChildrenSpec($fsmChildrenAllFwChildren) = chs
[convertChildren-any]
convertChildrenSpec($fsmChildrenAnyFwChildren) = chs

[convertStateNameSpec-single]
convertStateNameSpec($fsmStateName) = [toMcrlStateName($fsmStateName)]
[convertStateNameSpec-single-in-multiple]
convertStateNameSpec({$fsmStateName}) = [toMcrlStateName($fsmStateName)]
[convertStateNameSpec-multiple]
[ $dataExprs ] := convertStateNameSpec({$fsmStateNames})
===>
convertStateNameSpec({$fsmStateName, $fsmStateNames}) = [toMcrlStateName($fsmStateName), $dataExprs]

[inAnyState-empty]
inAnyState() = false
[inAnyState-one]
inAnyState($mid) = isStateCheck($mid, s1)
[inAnyState-many]
inAnyState($mid $mid+) = isStateCheck($mid, s1) || inAnyState($mid+)


[isStateCheck-2]
mid(#midHead #midTail) := toMcrlIsFunction(toMcrlStateName($mid1))
===>
isStateCheck($mid1, $mid2) = mid(#midHead #midTail)($mid2)

[isStateCheck-1]
isStateCheck($mid) = isStateCheck($mid, s)

[isCommand-c]
mid(#midHead #midTail) := toMcrlIsFunction(toMcrlCmdName($mid))
===>
isCommandCheck($mid) = mid(#midHead #midTail)(c)

[toMcrlIsFunc]
toMcrlIsFunction(mid(#midHead #midTail)) = mid(is #midHead #midTail)

[toMcrlState-1]
toMcrlStateName(mid(#midHead #midTail)) = concat(S_, mid(#midHead #midTail))

[toMcrlCmd-1]
toMcrlCmdName(mid(#midHead #midTail)) = concat(C_, mid(#midHead #midTail))



[generateSorts]
fsmGenerateSorts($fsmClass+) =
    PType = struct convertClassNamesToTypeConstrDecl(collectClasses($fsmClass+,));
    State = struct S_FSM_UNINITIALIZED ? isS_FSM_UNINITIALIZED | convertStateNamesToStateConstrDecl(collectStates($fsmClass+,));
    Command = struct convertActionNamesToCmdConstrDecl(collectCommands($fsmClass+,));

[convertStateNamesToSortDecl-single]
convertStateNamesToStateConstrDecl($mid) = toMcrlStateName($mid) ? toMcrlIsFunction(toMcrlStateName($mid))
[convertStateNamesToSortDecl-multi]
convertStateNamesToStateConstrDecl($mid $mid+) = convertStateNamesToStateConstrDecl($mid) | convertStateNamesToStateConstrDecl($mid+)

[convertActionNamesToSortDecl-single]
convertActionNamesToCmdConstrDecl($mid) = toMcrlCmdName($mid) ? toMcrlIsFunction(toMcrlCmdName($mid))
[convertActionNamesToSortDecl-multi]
convertActionNamesToCmdConstrDecl($mid $mid+) = convertActionNamesToCmdConstrDecl($mid) | convertActionNamesToCmdConstrDecl($mid+)

[convertClassNamesToSortDecl-single]
convertClassNamesToTypeConstrDecl($mid) = $mid ? toMcrlIsFunction($mid)
[convertClassNamesToTypeSortDecl-multi]
convertClassNamesToTypeConstrDecl($mid $mid+) = convertClassNamesToTypeConstrDecl($mid) | convertClassNamesToTypeConstrDecl($mid+)


[collect-class-definition]
collectClasses(class: $FWPART_$TOP$ $fsmClassName $fsmState+, $mid*) = addToSet($fsmClassName, $mid*)

[collect-class-exprall]
collectClasses($ALL$$fsmClassName, $mid*) = addToSet(concat($fsmClassName, _CLASS), $mid*)

[collect-class-exprany]
collectClasses($ANY$$fsmClassName, $mid*) = addToSet(concat($fsmClassName, _CLASS), $mid*)

[collect-command]
collectCommands(action: $fsmActionName $fsmStatement+, $mid*) = addToSet($fsmActionName, $mid*)

[collect-state]
collectStates(state: $fsmStateName $fsmWhenClause+ $fsmActionClause*, $mid*) = addToSet($fsmStateName, $mid*)

[collect-state-when]
collectStates(when ( $fsmExpr ) move_to $fsmStateName, $mid*) = addToSet($fsmStateName, $mid*)

[collect-state-statenamespec]
$mid*1 := collectStates({ $fsmStateNameSpecs* }, $mid*)
===>
collectStates( { $fsmStateName, $fsmStateNameSpecs*}, $mid*) = addToSet($fsmStateName, $mid*1)

[collect-state-statenamespec-1-element]
collectStates( $fsmStateName, $mid*) = addToSet($fsmStateName, $mid*)

[mcrl2GetPTypes-1]
mcrl2GetPTypes($procSpec+) =
    PType = struct mcrl2PTypesFromProcSpecs($procSpec+);

[mcrl2PTypesFromProcSpecs-one]
mcrl2PTypesFromProcSpecs(proc $procDecl+) =
    mcrl2PTypesFromProcDecls($procDecl+)

[mcrl2PTypesFromProcSpecs-many]
mcrl2PTypesFromProcSpecs(proc $procDecl+ $procSpec+) =
    mcrl2PTypesFromProcDecls($procDecl+) | mcrl2PTypesFromProcSpecs($procSpec+)

[mcrl2PTypesFromProcDecls-one-1]
mcrl2PTypesFromProcDecls($mid = $procExpr;) =
    convertClassNamesToTypeConstrDecl($mid)

[mcrl2PTypesFromProcDecls-many-1]
mcrl2PTypesFromProcDecls($mid = $procExpr; $procDecl+) =
    convertClassNamesToTypeConstrDecl($mid) | mcrl2PTypesFromProcDecls($procDecl+)

[mcrl2PTypesFromProcDecls-one-2]
mcrl2PTypesFromProcDecls($mid ( $idsDecls ) = $procExpr;) =
    convertClassNamesToTypeConstrDecl($mid)

[mcrl2PTypesFromProcDecls-many-2]
mcrl2PTypesFromProcDecls($mid ( $idsDecls ) = $procExpr; $procDecl+) =
    convertClassNamesToTypeConstrDecl($mid) | mcrl2PTypesFromProcDecls($procDecl+)

[addToSet-empty]
addToSet($mid,) = $mid
[addToSet-multisame]
addToSet($mid,$mid $mid*) = $mid $mid*
[addToSet-multidiff]
$mid != $mid1
===>
addToSet($mid,$mid1 $mid*) = $mid1 addToSet($mid, $mid*)
\end{listing}

\subsection{genericclauses.sdf}
\tiny
\begin{listing}[5]{1}
module genericclauses

imports basic/Comments
imports mcrlt

exports

context-free syntax

  insertGenericClauses(MId) -> ProcExpr

hiddens
variables
  "$fsmClassName" -> MId
\end{listing}

\subsection{genericclauses.asf}
\tiny
\begin{listing}[5]{1}
equations

[insertGeneric]
insertGenericClauses($fsmClassName) =
    (




    sum id:Id.(sum s1:State.(((isActPhase(phase)) && (is_child(id, chs)) &&
                             (pc(aArgs) == 0) && (initialized(chs))) ->
       rs(id, self, s1).
       move_phase(self, WhenPhase).
       $fsmClassName(self, parent, s, update_busy(id, false, update_state(id, s1, chs)), WhenPhase, reset(aArgs))))
    +

    sum id:Id.(sum s1:State.(((isActPhase(phase)) && (is_child(id, chs)) &&
           ((pc(aArgs) > 0) ||
            ((pc(aArgs) == -1) && (cq(aArgs) != [])))) ->
       rs(id, self, s1).
       $fsmClassName(self, parent, s,
                     update_busy(id,
                                 false,
                                 update_state(id, s1, chs)),
                     phase, aArgs))) +

    ((isActPhase(phase)) && (cq(aArgs) != []) && (!(initialized(chs)))) ->
        sc(self, id(head(cq(aArgs))), command(head(cq(aArgs)))).
            $fsmClassName(self, parent, s,
                          update_busy(id(head(cq(aArgs))), true, chs),
                          phase,
                          actArgs(tail(cq(aArgs)),
                                  (id(head(cq(aArgs))))|>(nrf(aArgs)), pc(aArgs), rsc(aArgs))) +

    ((isActPhase(phase)) && (cq(aArgs) != []) && initialized(chs)) ->
        sc(self, id(head(cq(aArgs))), command(head(cq(aArgs)))).
            $fsmClassName(self, parent, s,
                          update_busy(id(head(cq(aArgs))), true, chs),
                          phase,
                          actArgs(tail(cq(aArgs)), [], pc(aArgs), rsc(aArgs))) +


    sum id:Id.(sum s1:State.(
            ((isActPhase(phase)) && (cq(aArgs) == []) && (nrf(aArgs) != []) && (is_child(id, chs)) &&
             (!(initialized(chs)))) ->
                    rs(id, self, s1).
                    ((initialized(update_state(id,s1,chs))) ->
                         end_initialization(self).
                         $fsmClassName(self, parent, s, update_state(id,s1,chs), phase,
                                       actArgs(cq(aArgs), remove(id, nrf(aArgs)), -1, rsc(aArgs)
                                       )) <>
                         $fsmClassName(self, parent, s, update_state(id,s1,chs), phase,
                                       actArgs(cq(aArgs), remove(id, nrf(aArgs)), -1, rsc(aArgs)
                                       ))))) +

    ((isActPhase(phase)) && (cq(aArgs) == []) && (initialized(chs)) && (pc(aArgs) == -1)) ->
        move_phase(self, WhenPhase).
        $fsmClassName(self, parent, s, chs, WhenPhase, reset(aArgs))

    )
\end{listing}

\section{Wheel subsystem}
\tiny

\begin{listing}[5]{1}
class: $FWPART_$TOP$RPC_Wheel_CLASS
!panel: CMS_RPCfwSupervisor/CMS_RPCfwSupervisorRPC_Wheel.pnl
    state: OFF  !color: FwStateOKNotPhysics
        when ( $ANY$FwCHILDREN in_state ERROR )  move_to ERROR

        when ( $ANY$FwCHILDREN in_state RAMPING )  move_to RAMPING
        when ( $ALL$FwCHILDREN in_state STANDBY )  move_to STANDBY

        when ( $ALL$FwCHILDREN in_state ON )  move_to ON

        when ( ( $ALL$FwCHILDREN not_in_state OFF ) and
       ( $ANY$FwCHILDREN in_state STANDBY ) )  move_to STANDBY
        action: STANDBY !visible: 1
            do STANDBY $ALL$FwCHILDREN
        action: OFF !visible: 1
            do OFF $ALL$FwCHILDREN
        action: ON  !visible: 1
            do ON $ALL$FwCHILDREN
    state: STANDBY  !color: FwStateOKNotPhysics
        when ( $ANY$FwCHILDREN in_state ERROR )  move_to ERROR

        when ( $ANY$FwCHILDREN in_state RAMPING )  move_to RAMPING
        when ( $ALL$FwCHILDREN in_state ON )  move_to ON

        when ( $ANY$FwCHILDREN in_state OFF )  move_to OFF

        action: ON  !visible: 1
            do ON $ALL$FwCHILDREN
        action: OFF !visible: 1
            do OFF $ALL$FwCHILDREN
        action: STANDBY !visible: 1
            do STANDBY $ALL$FwCHILDREN
    state: ON   !color: FwStateOKPhysics
        when ( $ANY$FwCHILDREN in_state ERROR )  move_to ERROR

        when ( $ANY$FwCHILDREN in_state RAMPING)  move_to RAMPING
        when ( $ANY$FwCHILDREN in_state OFF )  move_to OFF

        when ( $ANY$FwCHILDREN in_state STANDBY )  move_to STANDBY

        action: STANDBY !visible: 1
            do STANDBY $ALL$FwCHILDREN
        action: OFF !visible: 1
            do OFF $ALL$FwCHILDREN
        action: ON  !visible: 1
            do ON $ALL$FwCHILDREN
    state: ERROR    !color: FwStateAttention3
        when ( ( $ANY$FwCHILDREN in_state RAMPING ) and
( $ALL$FwCHILDREN not_in_state ERROR ) )  move_to RAMPING
        when ( ( $ANY$FwCHILDREN in_state OFF ) and
( $ALL$FwCHILDREN not_in_state ERROR ) )  move_to OFF

        when ( $ALL$FwCHILDREN in_state ON )  move_to ON

        when ( ( $ANY$FwCHILDREN in_state STANDBY ) and
( $ALL$FwCHILDREN not_in_state ERROR ) )  move_to STANDBY

        action: ON  !visible: 1
            do ON $ALL$FwCHILDREN
        action: STANDBY !visible: 1
            do STANDBY $ALL$FwCHILDREN
        action: OFF !visible: 1
            do OFF $ALL$FwCHILDREN
    state: RAMPING  !color: FwStateAttention1
        when ( $ANY$FwCHILDREN in_state ERROR )  move_to ERROR
        when ( $ALL$FwCHILDREN in_state ON )  move_to ON
        when ( $ALL$FwCHILDREN in_state STANDBY )  move_to STANDBY
        when ( ( $ALL$FwCHILDREN not_in_state RAMPING ) and
       ( $ANY$FwCHILDREN in_state OFF ) )  move_to OFF
        when ( ( $ALL$FwCHILDREN not_in_state RAMPING ) and
       ( $ANY$FwCHILDREN in_state STANDBY ) )  move_to STANDBY
        action: STANDBY !visible: 1
            do STANDBY $ALL$FwCHILDREN
        action: OFF !visible: 1
            do OFF $ALL$FwCHILDREN
        action: ON  !visible: 1
            do ON $ALL$FwCHILDREN

\end{listing}

\section{Wheel translation}
\tiny
\begin{listing}[5]{1}
sort
    Phase = struct WhenPhase ?isWhenPhase | ActionPhase ?isActPhase;
    ActPhaseArgs = struct actArgs(cq: CommandQueue, nrf: IdList, pc: Int, rsc: Bool);
    Id = Nat;
    IdList = List(Id);
    Child = struct child(id:Id, state:State, ptype:PType, busy:Bool);
    Children = List(Child);
    ChildCommand = struct childcommand(id:Id, command:Command);
    CommandQueue = List(ChildCommand);
    PType = struct RPC_Wheel_CLASS ? isRPC_Wheel_CLASS;
    State = struct S_FSM_UNINITIALIZED ? isS_FSM_UNINITIALIZED | S_OFF ? isS_OFF | S_ERROR ? isS_ERROR |
                   S_RAMPING ? isS_RAMPING | S_STANDBY ? isS_STANDBY | S_ON ? isS_ON;
    Command = struct C_STANDBY ? isC_STANDBY | C_OFF ? isC_OFF | C_ON ? isC_ON;

act
    rc,sc,cc: Id # Id # Command;
    rs,ss,cs: Id # Id # State;
    move_state: Id # State;
    move_phase: Id # Phase;
    ignored_command: Id # Command;
    queue_messages: Id;
    enter_then_clause: Id;
    enter_else_clause: Id;
    skip_then_clause: Id;
    start_initialization: Id;
    end_initialization: Id;
    noop_statement: Id;

map
    in_state: Child # State -> Bool;
    in_any_of_states: Child # List(State) -> Bool;
    any_in_state: Children # List(State) -> Bool;
    all_in_state: Children # List(State) -> Bool;
    is_child: Id # Children -> Bool;
    filter_children: Children # PType -> Children;
    filter_children_accu: Children # PType # Children -> Children;
    send_command: Command # Children -> CommandQueue;
    update_state: Id # State # Children -> Children;
    update_busy: Id # Bool # Children -> Children;
    update_busy_all: Bool # Children -> Children;
    remove: Id # IdList -> IdList;
    initialized: Children -> Bool;
    children_to_ids: Children -> IdList;
    busy_children: Children -> Children;
    update_pc: ActPhaseArgs # Int -> ActPhaseArgs;
    reset: ActPhaseArgs -> ActPhaseArgs;

var
    cq: CommandQueue;
    chs,chs_accu: Children;
    ch: Child;
    id,id1: Id;
    ids: IdList;
    s,s1: State;
    sl: List(State);
    t,t1: PType;
    cmd: Command;
    b, b1, b2: Bool;
    pc, pc1: Int;

eqn
    in_state(child(id,s,t,b),s1) = s == s1;

    in_any_of_states(ch,[]) = false;
    in_any_of_states(ch,s|>sl) = in_state(ch,s) || in_any_of_states(ch,sl);

    any_in_state([], sl) = false;
    any_in_state(ch|> chs, sl) = in_any_of_states(ch,sl) || any_in_state(chs,sl);

    all_in_state([], sl) = true;
    all_in_state(ch|> chs, sl) = in_any_of_states(ch,sl) && all_in_state(chs,sl);

    is_child(id, []) = false;
    is_child(id, child(id1,s,t,b) |> chs) = id == id1 || is_child(id, chs);

    filter_children(chs, t) = filter_children_accu(chs, t, []);

    filter_children_accu([],t,chs_accu) = chs_accu;
    filter_children_accu(child(id,s,t1,b) |> chs, t, chs_accu) =
        if(t==t1,
           filter_children_accu(chs, t, child(id,s,t,b) |> chs_accu),
           filter_children_accu(chs, t, chs_accu));

    send_command(cmd, []) = [];
    send_command(cmd, child(id,s,t,b) |> chs) =
        childcommand(id,cmd) |> send_command(cmd,chs);

    update_state(id, s, []) = [];
    update_state(id, s, child(id1,s1,t,b) |> chs) =
        if(id==id1,
           child(id1,s,t,b) |> chs,
           child(id1,s1,t,b) |> update_state(id,s,chs));

    update_busy(id, b, []) = [];
    update_busy(id, b, child(id1,s,t,b1) |> chs) =
        if(id==id1,
           child(id1,s,t,b) |> chs,
           child(id1,s,t,b1) |> update_busy(id,b,chs));

    update_busy_all(b, []) = [];
    update_busy_all(b, child(id,s,t,b1) |> chs) = child(id,s,t,b) |> update_busy_all(b, chs);

    remove(id, []) = [];
    remove(id, id1 |> ids) =
      if (id == id1,
        ids,
        id1 |> remove(id, ids));

    initialized(chs) = !any_in_state(chs, [S_FSM_UNINITIALIZED]);

    children_to_ids([]) = [];
    children_to_ids(child(id,s,t,b) |> chs) = id |> children_to_ids(chs);

    busy_children([]) = [];
    busy_children(child(id,s,t,true) |> chs) = child(id,s,t,true) |> busy_children(chs);
    busy_children(child(id,s,t,false) |> chs) = busy_children(chs);

    update_pc(actArgs(cq, ids, pc, b), pc1) = actArgs(cq, ids, pc1, b);

    reset(actArgs(cq, ids, pc, b)) = actArgs([], [], 0, b);

proc RPC_Wheel_CLASS(self: Id, parent: Id, s: State, chs: Children, phase: Phase, aArgs: ActPhaseArgs) =
    (


    (
    ((isS_OFF(s)) && (isWhenPhase(phase)) &&
     (any_in_state(chs, [S_ERROR]))) ->
    move_state(self, S_ERROR).
    RPC_Wheel_CLASS(self, parent, S_ERROR, chs, phase, aArgs) <>
    ((
    ((isS_OFF(s)) && (isWhenPhase(phase)) &&
     (any_in_state(chs, [S_RAMPING]))) ->
    move_state(self, S_RAMPING).
    RPC_Wheel_CLASS(self, parent, S_RAMPING, chs, phase, aArgs) <>
    ((
    ((isS_OFF(s)) && (isWhenPhase(phase)) &&
     (all_in_state(chs, [S_STANDBY]))) ->
    move_state(self, S_STANDBY).
    RPC_Wheel_CLASS(self, parent, S_STANDBY, chs, phase, aArgs) <>
    ((
    ((isS_OFF(s)) && (isWhenPhase(phase)) &&
     (all_in_state(chs, [S_ON]))) ->
    move_state(self, S_ON).
    RPC_Wheel_CLASS(self, parent, S_ON, chs, phase, aArgs) <>
    ((
    ((isS_OFF(s)) && (isWhenPhase(phase)) &&
     ((!(any_in_state(chs, [S_OFF]))) && (any_in_state(chs, [S_STANDBY])))) ->
    move_state(self, S_STANDBY).
    RPC_Wheel_CLASS(self, parent, S_STANDBY, chs, phase, aArgs) <>
    ((
    (((isS_OFF(s)) && (isWhenPhase(phase))) ->
    ss(self, parent, s).
    move_phase(self, ActionPhase).
    RPC_Wheel_CLASS(self, parent, s, chs, ActionPhase, reset(aArgs)))
    ))
    ))
    ))
    ))
    ))
    ) +

    ((isS_OFF(s)) && (isActPhase(phase)) && (!(initialized(chs))) &&
     (pc(aArgs) == 0) && (nrf(aArgs) == [])) ->
        start_initialization(self).
        RPC_Wheel_CLASS(self, parent, s, chs, phase,
                  actArgs([], children_to_ids(chs), 0, rsc(aArgs))) <>

    ((initialized(chs)) ->
    (
    (((isS_OFF(s)) && (isActPhase(phase)) && (cq(aArgs) == []) && (pc(aArgs) == 0)) ->
        sum c:Command.(
           rc(parent, self, c).
           (isC_STANDBY(c) -> RPC_Wheel_CLASS(self, parent, s, chs, phase, update_pc(aArgs, 1)) <> (
           (isC_OFF(c) -> RPC_Wheel_CLASS(self, parent, s, chs, phase, update_pc(aArgs, 2)) <> (
           (isC_ON(c) -> RPC_Wheel_CLASS(self, parent, s, chs, phase, update_pc(aArgs, 3)) <> (
           ss(self, parent, s).
    ignored_command(self, c).
    RPC_Wheel_CLASS(self, parent, s, chs, phase, update_pc(aArgs, -1))))))))
    )) +

    (
    (((isS_OFF(s)) && (isActPhase(phase)) && (cq(aArgs) == [])) ->
    ((
    ((pc(aArgs) == 1) ->
         queue_messages(self).
         (RPC_Wheel_CLASS(self, parent, s, chs, phase,
                    actArgs(send_command(C_STANDBY,
                                         chs), [],  -1, rsc(aArgs)))))
    ))) +

    ((
    (((isS_OFF(s)) && (isActPhase(phase)) && (cq(aArgs) == [])) ->
    ((
    ((pc(aArgs) == 2) ->
         queue_messages(self).
         (RPC_Wheel_CLASS(self, parent, s, chs, phase,
                    actArgs(send_command(C_OFF,
                                         chs), [],  -1, rsc(aArgs)))))
    ))) +

    ((
    (((isS_OFF(s)) && (isActPhase(phase)) && (cq(aArgs) == [])) ->
    ((
    ((pc(aArgs) == 3) ->
         queue_messages(self).
         (RPC_Wheel_CLASS(self, parent, s, chs, phase,
                    actArgs(send_command(C_ON,
                                         chs), [],  -1, rsc(aArgs)))))
    ))) +

    (delta))))))
    ))

    ) +

    (

    (
    ((isS_STANDBY(s)) && (isWhenPhase(phase)) &&
     (any_in_state(chs, [S_ERROR]))) ->
    move_state(self, S_ERROR).
    RPC_Wheel_CLASS(self, parent, S_ERROR, chs, phase, aArgs) <>
    ((
    ((isS_STANDBY(s)) && (isWhenPhase(phase)) &&
     (any_in_state(chs, [S_RAMPING]))) ->
    move_state(self, S_RAMPING).
    RPC_Wheel_CLASS(self, parent, S_RAMPING, chs, phase, aArgs) <>
    ((
    ((isS_STANDBY(s)) && (isWhenPhase(phase)) &&
     (all_in_state(chs, [S_ON]))) ->
    move_state(self, S_ON).
    RPC_Wheel_CLASS(self, parent, S_ON, chs, phase, aArgs) <>
    ((
    ((isS_STANDBY(s)) && (isWhenPhase(phase)) &&
     (any_in_state(chs, [S_OFF]))) ->
    move_state(self, S_OFF).
    RPC_Wheel_CLASS(self, parent, S_OFF, chs, phase, aArgs) <>
    ((
    (((isS_STANDBY(s)) && (isWhenPhase(phase))) ->
    ss(self, parent, s).
    move_phase(self, ActionPhase).
    RPC_Wheel_CLASS(self, parent, s, chs, ActionPhase, reset(aArgs)))
    ))
    ))
    ))
    ))
    ) +

    ((isS_STANDBY(s)) && (isActPhase(phase)) && (!(initialized(chs))) &&
     (pc(aArgs) == 0) && (nrf(aArgs) == [])) ->
        start_initialization(self).
        RPC_Wheel_CLASS(self, parent, s, chs, phase,
                  actArgs([], children_to_ids(chs), 0, rsc(aArgs))) <>

    ((initialized(chs)) ->
    (
    (((isS_STANDBY(s)) && (isActPhase(phase)) && (cq(aArgs) == []) && (pc(aArgs) == 0)) ->
        sum c:Command.(
           rc(parent, self, c).
           (isC_ON(c) -> RPC_Wheel_CLASS(self, parent, s, chs, phase, update_pc(aArgs, 1)) <> (
           (isC_OFF(c) -> RPC_Wheel_CLASS(self, parent, s, chs, phase, update_pc(aArgs, 2)) <> (
           (isC_STANDBY(c) -> RPC_Wheel_CLASS(self, parent, s, chs, phase, update_pc(aArgs, 3)) <> (
           ss(self, parent, s).
    ignored_command(self, c).
    RPC_Wheel_CLASS(self, parent, s, chs, phase, update_pc(aArgs, -1))))))))
    )) +

    (
    (((isS_STANDBY(s)) && (isActPhase(phase)) && (cq(aArgs) == [])) ->
    ((
    ((pc(aArgs) == 1) ->
         queue_messages(self).
         (RPC_Wheel_CLASS(self, parent, s, chs, phase,
                    actArgs(send_command(C_ON,
                                         chs), [],  -1, rsc(aArgs)))))
    ))) +

    ((
    (((isS_STANDBY(s)) && (isActPhase(phase)) && (cq(aArgs) == [])) ->
    ((
    ((pc(aArgs) == 2) ->
         queue_messages(self).
         (RPC_Wheel_CLASS(self, parent, s, chs, phase,
                    actArgs(send_command(C_OFF,
                                         chs), [],  -1, rsc(aArgs)))))
    ))) +

    ((
    (((isS_STANDBY(s)) && (isActPhase(phase)) && (cq(aArgs) == [])) ->
    ((
    ((pc(aArgs) == 3) ->
         queue_messages(self).
         (RPC_Wheel_CLASS(self, parent, s, chs, phase,
                    actArgs(send_command(C_STANDBY,
                                         chs), [],  -1, rsc(aArgs)))))
    ))) +

    (delta))))))
    ))

    ) +

    (

    (
    ((isS_ON(s)) && (isWhenPhase(phase)) &&
     (any_in_state(chs, [S_ERROR]))) ->
    move_state(self, S_ERROR).
    RPC_Wheel_CLASS(self, parent, S_ERROR, chs, phase, aArgs) <>
    ((
    ((isS_ON(s)) && (isWhenPhase(phase)) &&
     (any_in_state(chs, [S_RAMPING]))) ->
    move_state(self, S_RAMPING).
    RPC_Wheel_CLASS(self, parent, S_RAMPING, chs, phase, aArgs) <>
    ((
    ((isS_ON(s)) && (isWhenPhase(phase)) &&
     (any_in_state(chs, [S_OFF]))) ->
    move_state(self, S_OFF).
    RPC_Wheel_CLASS(self, parent, S_OFF, chs, phase, aArgs) <>
    ((
    ((isS_ON(s)) && (isWhenPhase(phase)) &&
     (any_in_state(chs, [S_STANDBY]))) ->
    move_state(self, S_STANDBY).
    RPC_Wheel_CLASS(self, parent, S_STANDBY, chs, phase, aArgs) <>
    ((
    (((isS_ON(s)) && (isWhenPhase(phase))) ->
    ss(self, parent, s).
    move_phase(self, ActionPhase).
    RPC_Wheel_CLASS(self, parent, s, chs, ActionPhase, reset(aArgs)))
    ))
    ))
    ))
    ))
    ) +

    ((isS_ON(s)) && (isActPhase(phase)) && (!(initialized(chs))) &&
     (pc(aArgs) == 0) && (nrf(aArgs) == [])) ->
        start_initialization(self).
        RPC_Wheel_CLASS(self, parent, s, chs, phase,
                  actArgs([], children_to_ids(chs), 0, rsc(aArgs))) <>
    ((initialized(chs)) ->
    (
    (((isS_ON(s)) && (isActPhase(phase)) && (cq(aArgs) == []) && (pc(aArgs) == 0)) ->
        sum c:Command.(
           rc(parent, self, c).
           (isC_STANDBY(c) -> RPC_Wheel_CLASS(self, parent, s, chs, phase, update_pc(aArgs, 1)) <> (
           (isC_OFF(c) -> RPC_Wheel_CLASS(self, parent, s, chs, phase, update_pc(aArgs, 2)) <> (
           (isC_ON(c) -> RPC_Wheel_CLASS(self, parent, s, chs, phase, update_pc(aArgs, 3)) <> (
           ss(self, parent, s).
    ignored_command(self, c).
    RPC_Wheel_CLASS(self, parent, s, chs, phase, update_pc(aArgs, -1))))))))
    )) +
    (
    (((isS_ON(s)) && (isActPhase(phase)) && (cq(aArgs) == [])) ->
    ((
    ((pc(aArgs) == 1) ->
         queue_messages(self).
         (RPC_Wheel_CLASS(self, parent, s, chs, phase,
                    actArgs(send_command(C_STANDBY,
                                         chs), [],  -1, rsc(aArgs)))))
    ))) +

    ((
    (((isS_ON(s)) && (isActPhase(phase)) && (cq(aArgs) == [])) ->
    ((
    ((pc(aArgs) == 2) ->
         queue_messages(self).
         (RPC_Wheel_CLASS(self, parent, s, chs, phase,
                    actArgs(send_command(C_OFF,
                                         chs), [],  -1, rsc(aArgs)))))
    ))) +

    ((
    (((isS_ON(s)) && (isActPhase(phase)) && (cq(aArgs) == [])) ->
    ((
    ((pc(aArgs) == 3) ->
         queue_messages(self).
         (RPC_Wheel_CLASS(self, parent, s, chs, phase,
                    actArgs(send_command(C_ON,
                                         chs), [],  -1, rsc(aArgs)))))
    ))) +

    (delta))))))
    ))

    ) +

    (

    (
    ((isS_ERROR(s)) && (isWhenPhase(phase)) &&
     ((any_in_state(chs, [S_RAMPING])) && (!(any_in_state(chs, [S_ERROR]))))) ->
    move_state(self, S_RAMPING).
    RPC_Wheel_CLASS(self, parent, S_RAMPING, chs, phase, aArgs) <>
    ((
    ((isS_ERROR(s)) && (isWhenPhase(phase)) &&
     ((any_in_state(chs, [S_OFF])) && (!(any_in_state(chs, [S_ERROR]))))) ->
    move_state(self, S_OFF).
    RPC_Wheel_CLASS(self, parent, S_OFF, chs, phase, aArgs) <>
    ((
    ((isS_ERROR(s)) && (isWhenPhase(phase)) &&
     (all_in_state(chs, [S_ON]))) ->
    move_state(self, S_ON).
    RPC_Wheel_CLASS(self, parent, S_ON, chs, phase, aArgs) <>
    ((
    ((isS_ERROR(s)) && (isWhenPhase(phase)) &&
     ((any_in_state(chs, [S_STANDBY])) && (!(any_in_state(chs, [S_ERROR]))))) ->
    move_state(self, S_STANDBY).
    RPC_Wheel_CLASS(self, parent, S_STANDBY, chs, phase, aArgs) <>
    ((
    (((isS_ERROR(s)) && (isWhenPhase(phase))) ->
    ss(self, parent, s).
    move_phase(self, ActionPhase).
    RPC_Wheel_CLASS(self, parent, s, chs, ActionPhase, reset(aArgs)))
    ))
    ))
    ))
    ))
    ) +

    ((isS_ERROR(s)) && (isActPhase(phase)) && (!(initialized(chs))) &&
     (pc(aArgs) == 0) && (nrf(aArgs) == [])) ->
        start_initialization(self).
        RPC_Wheel_CLASS(self, parent, s, chs, phase,
                  actArgs([], children_to_ids(chs), 0, rsc(aArgs))) <>
    ((initialized(chs)) ->
    (
    (((isS_ERROR(s)) && (isActPhase(phase)) && (cq(aArgs) == []) && (pc(aArgs) == 0)) ->
        sum c:Command.(
           rc(parent, self, c).
           (isC_ON(c) -> RPC_Wheel_CLASS(self, parent, s, chs, phase, update_pc(aArgs, 1)) <> (
           (isC_STANDBY(c) -> RPC_Wheel_CLASS(self, parent, s, chs, phase, update_pc(aArgs, 2)) <> (
           (isC_OFF(c) -> RPC_Wheel_CLASS(self, parent, s, chs, phase, update_pc(aArgs, 3)) <> (
           ss(self, parent, s).
    ignored_command(self, c).
    RPC_Wheel_CLASS(self, parent, s, chs, phase, update_pc(aArgs, -1))))))))
    )) +
    (
    (((isS_ERROR(s)) && (isActPhase(phase)) && (cq(aArgs) == [])) ->
    ((
    ((pc(aArgs) == 1) ->
         queue_messages(self).
         (RPC_Wheel_CLASS(self, parent, s, chs, phase,
                    actArgs(send_command(C_ON,
                                         chs), [],  -1, rsc(aArgs)))))
    ))) +
    ((
    (((isS_ERROR(s)) && (isActPhase(phase)) && (cq(aArgs) == [])) ->
    ((
    ((pc(aArgs) == 2) ->
         queue_messages(self).
         (RPC_Wheel_CLASS(self, parent, s, chs, phase,
                    actArgs(send_command(C_STANDBY,
                                         chs), [],  -1, rsc(aArgs)))))
    ))) +
    ((
    (((isS_ERROR(s)) && (isActPhase(phase)) && (cq(aArgs) == [])) ->
    ((
    ((pc(aArgs) == 3) ->
         queue_messages(self).
         (RPC_Wheel_CLASS(self, parent, s, chs, phase,
                    actArgs(send_command(C_OFF,
                                         chs), [],  -1, rsc(aArgs)))))
    ))) +
    (delta))))))
    ))

    ) +

    (

    (
    ((isS_RAMPING(s)) && (isWhenPhase(phase)) &&
     (any_in_state(chs, [S_ERROR]))) ->
    move_state(self, S_ERROR).
    RPC_Wheel_CLASS(self, parent, S_ERROR, chs, phase, aArgs) <>
    ((
    ((isS_RAMPING(s)) && (isWhenPhase(phase)) &&
     (all_in_state(chs, [S_ON]))) ->
    move_state(self, S_ON).
    RPC_Wheel_CLASS(self, parent, S_ON, chs, phase, aArgs) <>
    ((
    ((isS_RAMPING(s)) && (isWhenPhase(phase)) &&
     (all_in_state(chs, [S_STANDBY]))) ->
    move_state(self, S_STANDBY).
    RPC_Wheel_CLASS(self, parent, S_STANDBY, chs, phase, aArgs) <>
    ((
    ((isS_RAMPING(s)) && (isWhenPhase(phase)) &&
     ((!(any_in_state(chs, [S_RAMPING]))) && (any_in_state(chs, [S_OFF])))) ->
    move_state(self, S_OFF).
    RPC_Wheel_CLASS(self, parent, S_OFF, chs, phase, aArgs) <>
    ((
    ((isS_RAMPING(s)) && (isWhenPhase(phase)) &&
     ((!(any_in_state(chs, [S_RAMPING]))) && (any_in_state(chs, [S_STANDBY])))) ->
    move_state(self, S_STANDBY).
    RPC_Wheel_CLASS(self, parent, S_STANDBY, chs, phase, aArgs) <>
    ((
    (((isS_RAMPING(s)) && (isWhenPhase(phase))) ->
    ss(self, parent, s).
    move_phase(self, ActionPhase).
    RPC_Wheel_CLASS(self, parent, s, chs, ActionPhase, reset(aArgs)))
    ))
    ))
    ))
    ))
    ))
    ) +


    ((isS_RAMPING(s)) && (isActPhase(phase)) && (!(initialized(chs))) &&
     (pc(aArgs) == 0) && (nrf(aArgs) == [])) ->
        start_initialization(self).
        RPC_Wheel_CLASS(self, parent, s, chs, phase,
                  actArgs([], children_to_ids(chs), 0, rsc(aArgs))) <>
    ((initialized(chs)) ->
    (
    (((isS_RAMPING(s)) && (isActPhase(phase)) && (cq(aArgs) == []) && (pc(aArgs) == 0)) ->
        sum c:Command.(
           rc(parent, self, c).
           (isC_STANDBY(c) -> RPC_Wheel_CLASS(self, parent, s, chs, phase, update_pc(aArgs, 1)) <> (
           (isC_OFF(c) -> RPC_Wheel_CLASS(self, parent, s, chs, phase, update_pc(aArgs, 2)) <> (
           (isC_ON(c) -> RPC_Wheel_CLASS(self, parent, s, chs, phase, update_pc(aArgs, 3)) <> (
           ss(self, parent, s).
    ignored_command(self, c).
    RPC_Wheel_CLASS(self, parent, s, chs, phase, update_pc(aArgs, -1))))))))
    )) +

    (
    (((isS_RAMPING(s)) && (isActPhase(phase)) && (cq(aArgs) == [])) ->
    ((
    ((pc(aArgs) == 1) ->
         queue_messages(self).
         (RPC_Wheel_CLASS(self, parent, s, chs, phase,
                    actArgs(send_command(C_STANDBY,
                                         chs), [],  -1, rsc(aArgs)))))
    ))) +

    ((
    (((isS_RAMPING(s)) && (isActPhase(phase)) && (cq(aArgs) == [])) ->
    ((
    ((pc(aArgs) == 2) ->
         queue_messages(self).
         (RPC_Wheel_CLASS(self, parent, s, chs, phase,
                    actArgs(send_command(C_OFF,
                                         chs), [],  -1, rsc(aArgs)))))
    ))) +

    ((
    (((isS_RAMPING(s)) && (isActPhase(phase)) && (cq(aArgs) == [])) ->
    ((
    ((pc(aArgs) == 3) ->
         queue_messages(self).
         (RPC_Wheel_CLASS(self, parent, s, chs, phase,
                    actArgs(send_command(C_ON,
                                         chs), [],  -1, rsc(aArgs)))))
    ))) +

    (delta))))))
    ))

    ) +

    (
    sum id:Id.(sum s1:State.(((isActPhase(phase)) && (is_child(id, chs)) &&
                             (pc(aArgs) == 0) && (initialized(chs))) ->
       rs(id, self, s1).
       move_phase(self, WhenPhase).
       RPC_Wheel_CLASS(self, parent, s, update_busy(id, false, update_state(id, s1, chs)), WhenPhase, reset(aArgs)))) +

    sum id:Id.(sum s1:State.(((isActPhase(phase)) && (is_child(id, chs)) &&
           ((pc(aArgs) > 0) ||
            ((pc(aArgs) == -1) && (cq(aArgs) != [])))) ->
       rs(id, self, s1).
       RPC_Wheel_CLASS(self, parent, s,
                     update_busy(id,
                                 false,
                                 update_state(id, s1, chs)),
                     phase, aArgs))) +

    ((isActPhase(phase)) && (cq(aArgs) != []) && (!(initialized(chs)))) ->
        sc(self, id(head(cq(aArgs))), command(head(cq(aArgs)))).
            RPC_Wheel_CLASS(self, parent, s,
                          update_busy(id(head(cq(aArgs))), true, chs),
                          phase,
                          actArgs(tail(cq(aArgs)),
                                  (id(head(cq(aArgs))))|>(nrf(aArgs)), pc(aArgs), rsc(aArgs))) +

    ((isActPhase(phase)) && (cq(aArgs) != []) && initialized(chs)) ->
        sc(self, id(head(cq(aArgs))), command(head(cq(aArgs)))).
            RPC_Wheel_CLASS(self, parent, s,
                          update_busy(id(head(cq(aArgs))), true, chs),
                          phase,
                          actArgs(tail(cq(aArgs)), [], pc(aArgs), rsc(aArgs))) +

    sum id:Id.(sum s1:State.(
            ((isActPhase(phase)) && (cq(aArgs) == []) && (nrf(aArgs) != []) && (is_child(id, chs)) &&
             (!(initialized(chs)))) ->
                    rs(id, self, s1).
                    ((initialized(update_state(id,s1,chs))) ->
                         end_initialization(self).
                         RPC_Wheel_CLASS(self, parent, s, update_state(id,s1,chs), phase,
                                       actArgs(cq(aArgs), remove(id, nrf(aArgs)), -1, rsc(aArgs)
                                       )) <>
                         RPC_Wheel_CLASS(self, parent, s, update_state(id,s1,chs), phase,
                                       actArgs(cq(aArgs), remove(id, nrf(aArgs)), -1, rsc(aArgs)
                                       ))))) +

    ((isActPhase(phase)) && (cq(aArgs) == []) && (initialized(chs)) && (pc(aArgs) == -1)) ->
        move_phase(self, WhenPhase).
        RPC_Wheel_CLASS(self, parent, s, chs, WhenPhase, reset(aArgs))

    );

init
    allow({cs, cc, move_state, move_phase, ignored_command,
           queue_messages, enter_then_clause, enter_else_clause,
           skip_then_clause, start_initialization, end_initialization,
           noop_statement},
    comm({rs|ss -> cs, rc|sc -> cc},
    RPC_Wheel_CLASS(1, 1, S_OFF, [],
                    ActionPhase,actArgs([], [], 0, false))));
\end{listing}
\fi
\end{document}